\documentclass[11pt,journal,onecolumn,draftclsnofoot]{IEEEtran}
\usepackage{epsfig}
\usepackage{graphicx}
\usepackage{amsmath}
\usepackage{amsfonts}
\usepackage{subfig}
\DeclareCaptionLabelFormat{figure1}{Figure #2}

\long\def\symbolfootnote[#1]#2{\begingroup%
\def\thefootnote{\fnsymbol{footnote}}\footnote[#1]{#2}\endgroup}

\begin{document}


\newcommand{\nk}[1]{{{#1}_n}}
\newcommand{\zk}[1]{{{#1}_0}}
\newcommand{\SDR}{{\text{SDR}}}
\newcommand{\SNR}{{\text{SNR}}}
\newcommand{\tSNR}{\SNR}
\newcommand{\trho}{\tilde\rho}
\newcommand{\comb}[2]{{#1 \choose #2}}
\newcommand{\uzer}{\underline{0}}
\newcommand{\uV}{\underline{V}}
\newcommand{\uA}{\underline{A}}
\newcommand{\uD}{\underline{D}}
\newcommand{\uv}{\underline{v}}
\newcommand{\uT}{\underline{T}}
\newcommand{\ut}{\underline{t}}
\newcommand{\ur}{\underline{r}}
\newcommand{\uR}{\underline{R}}
\newcommand{\uc}{\underline{c}}
\newcommand{\uC}{\underline{C}}
\newcommand{\ul}{\underline{l}}
\newcommand{\uL}{\underline{L}}
\newcommand{\uh}{\underline{h}}
\newcommand{\uH}{\underline{H}}
\newcommand{\ue}{\underline{e}}
\newcommand{\uE}{\underline{E}}
\newcommand{\uG}{\underline{G}}
\newcommand{\ug}{\underline{g}}
\newcommand{\uz}{\underline{z}}
\newcommand{\uZ}{\underline{Z}}
\newcommand{\uu}{\underline{u}}
\newcommand{\uU}{\underline{U}}
\newcommand{\uj}{\underline{j}}
\newcommand{\uJ}{\underline{J}}
\newcommand{\uX}{\underline{X}}
\newcommand{\ux}{\underline{x}}
\newcommand{\uY}{\underline{Y}}
\newcommand{\uy}{\underline{y}}
\newcommand{\uW}{\underline{W}}
\newcommand{\uw}{\underline{w}}
\newcommand{\uth}{\underline{\theta}}
\newcommand{\uTh}{\underline{\theta}}
\newcommand{\uph}{\underline{\phi}}
\newcommand{\ual}{\underline{\alpha}}
\newcommand{\uxi}{\underline{\xi}}
\newcommand{\us}{\underline{s}}
\newcommand{\uS}{\underline{S}}
\newcommand{\un}{\underline{n}}
\newcommand{\uN}{\underline{N}}
\newcommand{\up}{\underline{p}}
\newcommand{\uq}{\underline{q}}
\newcommand{\uf}{\underline{f}}
\newcommand{\ua}{\underline{a}}
\newcommand{\ub}{\underline{b}}
\newcommand{\uDelta}{\underline{\Delta}}

\newcommand{\cA}{{\cal A}}
\newcommand{\tcA}{\tilde{\cA}}
\newcommand{\cB}{{\cal B}}
\newcommand{\cC}{{\cal C}}
\newcommand{\cc}{{\cal c}}
\newcommand{\cD}{{\cal D}}
\newcommand{\cE}{{\cal E}}
\newcommand{\cF}{{\cal F}}
\newcommand{\cH}{{\cal H}}
\newcommand{\cI}{{\cal I}}
\newcommand{\cK}{{\cal K}}
\newcommand{\cL}{{\cal L}}
\newcommand{\cN}{{\cal N}}
\newcommand{\cP}{{\cal P}}
\newcommand{\cQ}{{\cal Q}}
\newcommand{\cR}{{\cal R}}
\newcommand{\cS}{{\cal S}}
\newcommand{\cs}{{\cal s}}
\newcommand{\cT}{{\cal T}}
\newcommand{\ct}{{\cal t}}
\newcommand{\cU}{{\cal U}}
\newcommand{\cV}{{\cal V}}
\newcommand{\cW}{{\cal W}}
\newcommand{\cX}{{\cal X}}
\newcommand{\cx}{{\cal x}}
\newcommand{\cY}{{\cal Y}}
\newcommand{\cy}{{\cal y}}
\newcommand{\cZ}{{\cal Z}}

\newcommand{\Nu}{{\mathcal V}}

\newcommand{\tA}{\tilde{A}}
\newcommand{\tE}{\tilde{E}}
\newcommand{\tZ}{\tilde{Z}}
\newcommand{\tz}{\tilde{z}}
\newcommand{\tQ}{\tilde{Q}}
\newcommand{\tR}{\tilde{R}}
\newcommand{\hX}{\hat{X}}
\newcommand{\hY}{\hat{Y}}
\newcommand{\hZ}{\hat{Z}}
\newcommand{\huX}{\hat{\uX}}
\newcommand{\huY}{\hat{\uY}}
\newcommand{\huZ}{\hat{\uZ}}
\newcommand{\inthalf}[1]{\int_{-\frac{1}{2}}^{\frac{1}{2}} #1 df}
\newcommand{\indp}{\underline{\; \| \;}}
\newcommand{\diag}{\mbox{diag}}
\newcommand{\sumk}{\sum_{k=1}^{K}}
\newcommand{\beq}[1]{\begin{equation}\label{#1}}
\newcommand{\eeq}{\end{equation}}
\newcommand{\req}[1]{(\ref{#1})}
\newcommand{\beqn}[1]{\begin{eqnarray}\label{#1}}
\newcommand{\eeqn}{\end{eqnarray}}
\newcommand{\limn}{\lim_{n \rightarrow \infty}}
\newcommand{\limN}{\lim_{N \rightarrow \infty}}
\newcommand{\limr}{\lim_{r \rightarrow \infty}}
\newcommand{\limd}{\lim_{\delta \rightarrow \infty}}
\newcommand{\limM}{\lim_{M \rightarrow \infty}}
\newcommand{\limsupn}{\limsup_{n \rightarrow \infty}}
\newcommand{\imii}{\int_{-\infty}^{\infty}}
\newcommand{\imix}{\int_{-\infty}^x}
\newcommand{\ioi}{\int_o^\infty}
\newcommand{\bphi}{\mbox{\boldmath \begin{math}\phi\end{math}}}
\newcommand{\bth}{\mbox{\boldmath \begin{math}\theta\end{math}}}
\newcommand{\bhth}{\mbox{\boldmath \begin{math}\hat{\theta}\end{math}}}
\newcommand{\bg}{\mbox{\boldmath \begin{math}g\end{math}}}
\newcommand{\bA}{{\bf A}}
\newcommand{\ba}{{\bf a}}
\newcommand{\bB}{{\bf B}}
\newcommand{\bb}{{\bf b}}
\newcommand{\bc}{{\bf c}}
\newcommand{\bd}{{\bf d}}
\newcommand{\bD}{{\bf D}}
\newcommand{\bE}{{\bf E}}
\newcommand{\bff}{{\bf f}}
\newcommand{\bG}{{\bf G}}
\newcommand{\bW}{{\bf W}}
\newcommand{\bM}{{\bf M}}
\newcommand{\bi}{{\bf i}}
\newcommand{\bl}{{\bf l}}
\newcommand{\bm}{{\bf m}}
\newcommand{\bn}{{\bf n}}
\newcommand{\bp}{{\bf p}}
\newcommand{\bs}{{\bf s}}
\newcommand{\bt}{{\bf t}}
\newcommand{\bu}{{\bf u}}
\newcommand{\bx}{{\bf x}}
\newcommand{\by}{{\bf y}}
\newcommand{\bz}{{\bf z}}
\newcommand{\bC}{{\bf C}}
\newcommand{\bI}{{\bf I}}
\newcommand{\bJ}{{\bf J}}
\newcommand{\bN}{{\bf N}}
\newcommand{\bS}{{\bf S}}
\newcommand{\bT}{{\bf T}}
\newcommand{\bU}{{\bf U}}
\newcommand{\bV}{{\bf V}}
\newcommand{\bv}{{\bf v}}
\newcommand{\bw}{{\bf w}}
\newcommand{\bX}{{\bf X}}
\newcommand{\bY}{{\bf Y}}
\newcommand{\bZ}{{\bf Z}}
\newcommand{\oI}{\overline{I}}
\newcommand{\oD}{\overline{D}}
\newcommand{\oh}{\overline{h}}
\newcommand{\oV}{\overline{V}}
\newcommand{\oR}{\overline{R}}
\newcommand{\oH}{\overline{H}}
\newcommand{\ol}{\overline{l}}
\newcommand{\E}{{\cal E}_d}
\newcommand{\thref}[1]{Theorem \ref{#1}}
\newcommand{\figref}[1]{Figure \ref{#1}}
\newcommand{\secref}[1]{Section \ref{#1}}

\newcommand{\round}{\mathop{\mathrm{round}}}
\newcommand{\var}{\mathop{\mathrm{Var}}}

\newcommand{\ej}{{e^{j2\pi f}}}

\newcommand{\yesindent}{\hspace*{\parindent}}   
\newcommand{\pderiv}[2]{\frac{ \partial {#1}}{ \partial {#2}}}
\newcommand{\overr}[2]{\left({\begin{array}{l}#1\\#2\end{array}}\right)}
\newcommand{\Ddef}{\stackrel{\Delta}{=}}

\pagestyle{plain}
\newcommand{\bQ}{{\bf Q}}
\newcommand{\bq}{{\bf q}}
\newcommand{\el}{\ell}
\newcommand{\linf}{{\el\rightarrow\infty}}
\renewcommand{\thesection}{\Roman{section}}

\newtheorem{theorem}{Theorem}
\newtheorem{prop}{Proposition}
\newtheorem{cor}{Corollary}
\newtheorem{lemma}{Lemma}
\newtheorem{conj}{Conjucture}
\newtheorem{assume}{Assumption}
\newtheorem{definition}{Definition}

\newcommand{\ccaption}{\caption*~{Figure~\thefigure: }}
\newcommand{\tcaption}{\caption*~{Table~\thetable: }}


\title{Analog Matching of Colored Sources to Colored Channels
\symbolfootnote[2]{Parts of this work were presented at ISIT2006,
Seattle, WA, July 2006 and at ISIT2007, Nice, France, July 2007.
This work was supported by the Israeli Science Foundation (ISF)
under grant \# 1259/07, and by the Advanced Communication Center
(ACC). The first author was also supported by a fellowship of the
Yitzhak and Chaya Weinstein Research Institute for Signal Processing
at Tel Aviv University. }}

\author{Yuval Kochman and Ram Zamir \\
Dept. Electrical Engineering - Systems, Tel Aviv University }

\maketitle

\begin{abstract}
Analog (uncoded) transmission provides a simple and robust scheme for
communicating a Gaussian source over a Gaussian channel under the
mean squared error (MSE) distortion measure. Unfortunately, its
performance is usually inferior to the all-digital, separation-based
source-channel coding solution, which requires exact knowledge of
the channel at the encoder. The loss comes from the fact that except
for very special cases, e.g. white source and channel of matching
bandwidth (BW), it is impossible to achieve perfect matching of
source to channel and channel to source by linear means. We show
that by combining prediction and modulo-lattice operations, it is possible to
match any colored Gaussian source to any colored Gaussian noise channel (of possibly different BW),
hence achieve Shannon's optimum attainable performance $R(D)=C$.
Furthermore, when the source and channel BWs are equal (but
otherwise their spectra are arbitrary), this scheme is asymptotically
robust in the sense that for high signal-to-noise ratio a single 
encoder (independent of the noise variance) achieves the optimum performance. 
The derivation is based upon a
recent modulo-lattice modulation scheme for transmitting a Wyner-Ziv
source over a dirty-paper channel.
\end{abstract}

\vspace{5mm} \textbf{keywords:} joint source/channel coding, analog
transmission, Wyner-Ziv problem, writing on dirty paper, modulo
lattice modulation, MMSE estimation, prediction, unknown SNR,
broadcast channel, ISI channel, bandwidth expansion/reduction.

\section{Introduction} \label{Sec_Intro}


Digital transmission of analog sources relies, at least from a
theoretical point of view, on Shannon's source-channel separation
principle.
Being both optimal and easy to implement, digital techniques replace
today traditional analog communication even in areas like voice
telephony, radio and television. This trend ignores, however, the
fact that the separation principle does not hold for communication
networks, and in particular for broadcast channels and compound
channels \cite{CoverBook,Ziv70,TrottITW96}. Indeed, due to both
practical and theoretical reasons, {\em joint} source-channel coding
and hybrid digital-analog schemes are constantly receiving attention
of researchers in the academia and the industry.

In this work we consider transmission under the mean-squared error (MSE) distortion criterion, of a general
stationary Gaussian source over a power-constrained channel with inter-symbol interference (ISI), i.e. the transmitted signal is passed through some linear filter, and additive white Gaussian
noise (AWGN). \footnote{It turns out, that for the purpose of analysis it is more convenient to use a colored-noise channel model rather than an ISI one; this is deferred to \secref{Sec_Prel}.} 

Shannon's joint source-channel coding theorem implies that the
optimal (i.e., minimum distortion) performance $D^{opt}$ is given by
\begin{equation}
\label{OPTA1} R\left(D^{opt}\right)=C,
\end{equation}
where $R(D)$  
is the rate-distortion function of the source at MSE
distortion $D$, and
$C = C(P)$  
is the channel capacity at power-constraint
$P$, both given by the well-known water-filling solutions
\cite{CoverBook}.
By Shannon's separation principle, (\ref{OPTA1}) can be achieved by
a system consisting of source and channel coding schemes. This
system usually requires large delay and complex digital codes. An additional
serious drawback of the all-digital system is that it suffers
from a ``threshold effect'': if the channel noise turns out to be
higher than expected, then the reconstruction will suffer from very
large distortion, while if the channel has lower noise than
expected, then there is no improvement in the distortion
\cite{Ziv70,TrottITW96,ChenWornell98}.

In contrast, analog communication techniques (like amplitude or
frequency modulation \cite{CouchBook}) are not sensitive to exact
channel knowledge at the transmitter. Moreover, in spite of their
low complexity and delay, they are sometimes optimal: if we are
allowed one channel use per source sample, and the source and noise are white
(i.e. have i.i.d. samples), then a ``single-letter'' coding scheme
achieves the optimum performance of (\ref{OPTA1}), see e.g.
\cite{Goblick65}. In this scheme, the transmitter consists of
multiplication by a constant factor that adjusts the source to the
power constraint $P$, so it is independent of the channel
parameters. Only the receiver needs to know the power of the noise in order to
optimally estimate the source from the noisy channel output (by
multiplying by the ``Wiener coefficient'').

For the case of {\em colored} sources and channels, however, such a
simple solution is not available, as single-letter codes are only
optimal in very special scenarios \cite{ToCode}. A particular case is when the
channel bandwidth is
not equal to the source bandwidth, but otherwise they are
white (i.e., a white source is sent through an AWGN channel with some average number of channel uses per source
sample). As it turns out, even if we consider more
general linear transmission schemes, \cite{BergerOptimalPAM}, still
(\ref{OPTA1}) is not achievable in the general colored
case.
How far do we need to deviate from ``analog'' transmission in order
to achieve optimal performance in the colored case? More
importantly, can we still achieve full robustness?


In this work we propose and investigate a \emph{semi-analog} transmission
scheme. This scheme achieves the optimum performance of
(\ref{OPTA1}) for {\em any} colored source and channel pair without explicit digital coding, hence
we call it the \emph{Analog Matching} (AM) scheme. Furthermore, for the
matching bandwidth case ($B_C = B_S$, but arbitrary source and channel spectra), we show that the Analog
Matching transmitter is \emph{asymptotically robust} in the high
signal-to-noise ratio (SNR) regime, in the sense that it becomes
invariant to the variance of the channel noise. Thus, in this
regime, the perfect SNR-invariant matching property of white sources
and channels \cite{Goblick65} generalizes to the equal-BW colored
case.


Previous work on joint source/channel coding for the
BW-mismatch/colored setting mostly consists of hybrid digital analog
(HDA) solutions, which involve splitting the source or channel into
frequency bands, or using a superposition of encoders (see
\cite{ShamaiVerduZamir,MittalPhamdo,TzvikaBroadcast,Puri2005,NarayananCaireReport}
and references therein), mostly for the cases of bandwidth expansion
($B_C > B_S$) and bandwidth compression ($B_C < B_S$) with white
spectra. Most of these solutions, explicitly or
implicitly, allocate different power and bandwidth resources to
analog and digital source representations, thus they still employ
full \emph{digital coding}. 
Other works \cite{ChenWornell98,Vaishampayan2003} treat
bandwidth expansion by mapping each source sample to a sequence of
channel inputs independently; by the scalar nature of these mappings, they do not aim at optimal performance.

In contrast to HDA solutions, the AM scheme treats the
source and channel in the \emph{time domain},  using linear {\em prediction}, thus it also has the potential of shorter delay. Furthermore, it does not involve
any quantization of the source or digital channel code, but rather it applies {\em modulo-lattice
arithmetic} to analog signals. This modulation allows to take advantage of side information - here based on prediction - while keeping the analog nature of transmission.

\begin{table}[h]
\centering
\begin{tabular}{| c || c | c |}
\hline
Problem & Conventional prediction & Side-information based solution \\
\hline\hline
Source coding & DPCM compression & WZ video coding \\
\hline
Channel coding & FFE-DFE receiver & Dirty-paper coding = precoding \\
\hline
Joint source-channel coding & Does not exist & Analog matching \\ 
\hline
\end{tabular}
\tcaption{Information-Theoretic time-domain solutions to colored
Gaussian source and channel problems.} \label{si_problems_table}
\end{table}

Table~\ref{si_problems_table} demonstrates
the place of the Analog Matching scheme within
information-theoretic time-domain schemes. For the separate colored Gaussian source and channel problems, digital coding schemes, based upon
the combination of prediction and memoryless codebooks, are optimal: differential
pulse code modulation (DPCM) in source coding (see \cite{Jayant84} for basic properties and \cite{ZamirKochmanErezDPCM-IT} for optimality), and feed-forward-equalizer / decision-feedback-equalizer (FFE-DFE) receiver in channel coding (see \cite{CDEF-MMSE-DFE}). \footnote{In the high-rate limit it is easy to see the role of prediction: the rate-distortion function amounts to that of the white source innovations process, while the channel capacity is the additive white Gaussian noise channel capacity with the noise replaced by its innovations process only. We stick to this limit in the introduction; for general rates, see \secref{Sec_Prel}.} 

The optimality of DPCM hinges on prediction being performed using the reconstruction rather than the source itself. \footnote{Extracting the innovations of the un-quantized source is sometimes called ``D$^*$PCM'' and is known to be strictly inferior to DPCM; see \cite{Jayant84,ZamirKochmanErezDPCM-IT}.} Identical predictors, with equal outputs, are employed at the encoder and at the decoder. An alternative approach, advocated for low-complexity encoding, is ``Wyner-Ziv video coding''  (see e.g.
\cite{PuriRamchandranPRISM}). In this approach, prediction is performed at the decoder only and is treated as decoder side-information \cite{WynerZiv76}. In the context of the AM scheme, however, decoder-only prediction is not an option but a must: since no quantization is used, but rather the reconstruction error is generated by the channel, the encoder does not have access to the error and the side-information approach must be taken. 

In the channel counterpart, the FFE-DFE receiver cancels the effect of past channel inputs by filtering past decoder decisions (assumed to be equal to these inputs). In order to avoid error propagation, sometimes precoding \cite{Tomlinson}, where the filter is moved to the encoder, is preferred; this can be seen as a form of \emph{dirty-paper coding} \cite{Costa83} , where the filter output plays the role of encoder side-information. Again, the AM scheme must use the ``encoder side-information'' variant: if no channel code is used, then the decoder cannot make digital decisions regarding past channel inputs, so virtually it has no access to these inputs.    

\begin{figure*}[!t]
{\centering
\input{cualitative.pstex_t}}
\ccaption{Workings of the AM scheme in the high-SNR limit. The source is assumed to have an auto-regressive (AR) model. $\mod \Lambda$ is the modulo-lattice operation.}
\label{cualitative_fig}
\end{figure*}

To summarize, the AM scheme uses source prediction at the decoder, and channel prediction at the encoder, and then treats the predictor outputs as Wyner-Ziv and dirty-paper side-information, respectively; see \figref{cualitative_fig}. Digital solutions to these side-information problems rely on binning, which may also be materialized in a structured (lattice) way \cite{RamiShamaiUriLattices}. AM treats these two side-information problems jointly using \emph{modulo-lattice modulation} (MLM) , an approach proposed recently for joint Wyner-Ziv and dirty-paper coding \cite{JointWZ-WDP}. However, combining these pieces turns out to be a non-trivial task. The interaction of filters with high-dimensional lattice codes raises technical difficulties which are solved in the sequel.   

The rest of the paper is organized as follows: We start in
\secref{Sec_Prel} by bringing
preliminaries regarding sources and channels with memory, as well as
modulo-lattice modulation and side-information problems. In
\secref{Sec_Matching} we prove the optimality of the Analog Matching
scheme. In \secref{Sec_Universal} we analyze the scheme performance
for unknown SNR, and prove its asymptotic robustness. Finally,
\secref{conclusion} discusses applications of AM, and is advantage relative to other approaches (e.g. HDA) in terms of delay.

\section{Formulation and Preliminaries} \label{Sec_Prel}

In \secref{Sub_Problem} we formally present the problem.
In the rest of the section we bring preliminaries necessary for the rest of the
paper. In Sections~\ref{Sub_Pred} to \ref{Sub_DPCM} we present
results connecting the Gaussian-quadratic rate-distortion function (RDF) and the Gaussian
channel capacity to prediction, mostly following
\cite{ZamirKochmanErezDPCM-IT}. In sections~\ref{Sub_lattice} and
\ref{Sub_joint} we discuss lattices and their application to joint
source/channel coding with side information, mostly following
\cite{JointWZ-WDP}.

\subsection{Problem Formulation} \label{Sub_Problem}

\begin{figure}[t]
\centering
\input{colored_model.pstex_t}
\ccaption{Colored Gaussian joint source/channel setting.} \label
{fig_Setting}
\end{figure}

\figref{fig_Setting} demonstrates the setting we consider in this
paper. The source $S_n$ is zero-mean stationary Gaussian, with spectrum $S_S(\ej)$. As for the channel, for the purpose of the analysis to follow we break away with the ISI model discussed in the introduction, and use a colored noise model:  \footnote{The transition between the
two models is straightforward using a suitable front-end filter
at the receiver (provided that the ISI filter is invertible).}
\begin{equation}
\label{channel} Y_n = X_n + Z_n ,
\end{equation}
where $X_n$ and $Y_n$ are the channel input and output, $Z_n$ is zero-mean additive stationary Gaussian
noise with spectrum $S_Z(\ej)$, assumed to be finite for all $2|f| \leq B_C$ and infinite otherwise. 
The channel input $X_n$ needs to satisfy the power constraint $\var\{X_n\}\leq P$, and the distortion of the reconstruction $\hat S_n$ is given by $D=\var\{\hat S_n - S_n\}$.

\subsection{Spectral Decomposition and Prediction} \label{Sub_Pred}

Let $A_n$ be a zero-mean discrete-time stationary process, with power  spectrum $S_A(\ej)$.
The Paley-Wiener condition is given by
\cite{VanTrees68}: \beq{Paley} \left| \inthalf{\log
\Bigl(S_A(\ej)\Bigr)} \right| < \infty \ \ , \eeq where here and in the sequel logarithms are taken to the natural base. It holds
for example if the spectrum $S_A(\ej)$ is bounded away from zero.
Whenever the Paley-Wiener condition holds, the spectrum has a
spectral decomposition: \beq{spectral_decomposition}
S_A(\ej)=\left. Q(z)Q^*\left(\frac{1}{z^*}\right)\right|_{z=j2\pi f}
P_e\Bigl(S_A\Bigr) \ \ , \eeq where $Q(z)$ is a monic causal filter,
and the entropy-power $P_e\left(S_A\right)$ of the spectrum is defined
by: \beq{entropy_power} P_e(S_A) \Ddef P_e\Bigl( S_A(\ej)\Bigr) =
\exp{\int_{-\frac{1}{2}}^{\frac{1}{2}}\log \Bigl( S_A(\ej)\Bigr) df} \
\ . \eeq The \emph{optimal predictor} of the process $A_n$ from its infinite past is \beq{optimal_predictor} P(z) = 1 - Q^{-1}(z) \ \ , \eeq a filter with
an impulse response satisfying $p_n=0$ for all $n\leq 0$. The
prediction mean squared error (MSE) is equal to the entropy power of the process:
\beq{noiseless_prediction}  \var\{A_n|A_{-\infty}^{n-1}\} = P_e(S_A) \
\ . \eeq The prediction error process can
serve as a white innovations process for AR representation of the
process. 

We define the \emph{prediction gain} of a spectrum $S_A(\ej)$ as:
\beq{pred_gain} \Gamma(S_A) \Ddef \Gamma\Bigl(S_A(\ej)\Bigr) \Ddef \frac
{\inthalf{S_A(\ej)} }{P_e\left(S_A\right)} =
\frac{\var\{A_n\}}{\var\{A_n|A_{-\infty}^{n-1}\}} \geq 1 \ \ , \eeq
where the gain equals one if and only if the spectrum is white, i.e.
fixed over all frequencies $|f|\leq\frac{1}{2}$. A case of special
interest, is where the process is band-limited such that $S_A(\ej)=0 \
\forall |f|>\frac{B}{2}$ where $B<1$. In that case, the Paley-Wiener condition \eqref{Paley}
does not hold and the prediction gain is infinite. We re-define,
then, the prediction gain of a process band-limited to $B$ as the
gain of the process downsampled by $\frac{1}{B}$, i.e., \footnote{A similar definition can be made for more general cases, e.g., when the signal is band-limited to some band which does not start at zero frequency.}
\beq{pred_gain_rho} \Gamma(S) = \frac
{\frac{1}{B} \int_{-\frac{B}{2}}^{\frac{B}{2}} S_A(\ej) df}{\exp \left[ \frac{1}{B}
\int_{-\frac{B}{2}}^{\frac{B}{2}} \log \Bigl(S_A(\ej)\Bigr) df\right]} \ \ .
\eeq

We will use in the sequel prediction from a noisy version of a
process: Suppose that $C_n=A_n+W_n$, with $W_n$ additive white with power
$\theta$. Then it can be shown that the noisy prediction error has variance (see e.g. \cite{ZamirKochmanErezDPCM-IT}):
\beq{noisy_prediction} \var\{A_n|C_{-\infty}^{n-1}\} =
P_e(S_A+\theta)-\theta \ \ . \eeq Note that for any $\theta>0$, the
spectrum $S_A(\ej)+\theta$ obeys \eqref{Paley}, so that the
conditional variance is non-zero even if $A_n$ is band-limited. In the
case $\theta=0$, \eqref{noisy_prediction} collapses to
\eqref{noiseless_prediction}.

\subsection{Water-Filling Solutions and the Shannon Bounds}
\label{Sub_SLB}

The RDF for a Gaussian source with
spectrum $S_S(\ej)$ under an MSE distortion measure is given by:
\beq{RDF} R(D) = \frac{1}{2}\inthalf{\log\frac{S_S(\ej)}{D(\ej)}} \
\ ,\eeq where the \emph{distortion spectrum} $D(\ej)$ is given by
the reverse water-filling solution: $D(\ej)=\min\Bigl(
\theta_S,S(\ej)\Bigr)$ with the \emph{water level} $\theta_S$ set by
the distortion level $D$:
\[D = \int_{-1/2}^{1/2} D(\ej) df \ \ \ .\]
The \emph{Shannon lower bound} (SLB) for the RDF of a source
band-limited to $B_S$ is given by: \beq{SLB} R(D) \geq \frac{B_S}{2}
\log\frac{\SDR}{\Gamma_S} \Ddef R_{SLB}(D) \ \ , \eeq where SDR, the
\emph{signal-to-distortion ratio}, is defined as: \beq{SDR} \SDR \Ddef
\frac{\var\{S_n\}}{D} \eeq and $\Gamma_S\Ddef\Gamma(S_S)$ is the
source prediction gain \eqref{pred_gain_rho}. This bound is tight
for a Gaussian source whenever the distortion level $D$ is low
enough such that $D<B_S \min_{|f|\leq B_S} S_S(\ej)$, and consequently
$D(\ej)=\theta_S=\frac{D}{B_S}$ for all $|f|<B_S$. Note that the bound reflects a coding rate gain of $B_S/2 \log(\Gamma_S)$ with respect to the RDF of a white Gaussian source.

The capacity of the colored channel \eqref{channel} where the noise $Z_n$ has spectrum $S_Z(\ej)$,
bandlimited to $B_C$, is given by: \beq{C} C
= \int_{-\frac{B_C}{2}}^{\frac{B_C}{2}} {\log\left(1+\frac{P(\ej)}{S_Z(\ej)}\right)} \ \ , \eeq
where the \emph{optimum channel input spectrum} $P(\ej)$ is given by the
water-filling solution: $P(\ej)=\max\Bigl(
\theta_C-S_Z(\ej),0\Bigr)$ inside the band, with the \emph{water level} $\theta_C$
set by the power constraint $P$:
\[P = \int_{-B_C/2}^{B_C/2} P(\ej) df \ \ \ .\] The \emph{Shannon upper bound} (SUB) for the channel capacity is
given by: \beq{SUB} C \leq \frac{B_C}{2} \log \left[\Gamma_C \cdot
\left(1+\tSNR\right)\right] \Ddef C_{SUB} \ \ , \eeq where SNR, the \emph{signal-to-noise ratio}, is defined as: \beq{tSNR} \SNR
\Ddef \frac {P}{N} \Ddef \frac{P}{\int_{-B_C/2}^{B_C/2}{S_Z(\ej)}df}\ \ , \eeq and $\Gamma_C\Ddef\Gamma(S_Z)$ is the channel prediction gain \eqref{pred_gain}. 
The bound is
tight for a Gaussian channel whenever the SNR is high
enough such that $P \geq B_C \max_{|f|\leq B_C} S_Z(\ej) - N$
and consequently $S_Z(\ej)+P(\ej)=\theta_C=\frac{P+N}{B_C}$. Note that the bound reflects a coding rate gain of $B_C/2 \log(\Gamma_C)$ with respect to the AWGN channel capacity.

We now connect the capacity and RDF expressions. In terms of the SDR and SNR defined above, and denoting the inverse of the RDF by $R^{-1}(\cdot)$, the optimal performance \eqref{OPTA1} becomes: \beq{OPTA} \SDR^{opt}\Ddef \frac{\var\{S_n\}}{R^{-1}(C(\tSNR))} \ \
. \eeq Let the \emph{bandwidth ratio} be \beq{rho} \rho \Ddef
\frac{B_C}{B_S} \ \ . \eeq 
Combining \eqref{SLB} with \eqref{SUB}, we have the following
asymptotically tight upper bound on the Shannon optimum performance. It shows that the prediction gains product $\Gamma_S\Gamma_C$ gives the total SDR gain relative to the case where the source and channel spectra are white.

\vspace{5mm}
\begin{prop} \label{RDeqC} 
\[ \frac{SDR^{opt}}{(1+\tSNR)^\rho} \leq
\Gamma_S \Gamma_C  \ \ , \] with equality if and
only if the SLB and SUB both hold with equality. \footnote{The SLB and
SUB never strictly hold with equality if $S_S(\ej)$ is not bounded away from zero,
or $S_Z(\ej)$ is not everywhere finite. However, they do hold
asymptotically in the high-SNR limit, if these spectra satisfy the Paley-Wiener
condition \eqref{Paley}.} Furthermore, if the source and the channel noise both satisfy the Paley-Wiener condition \eqref{Paley} inside their respective bandwidths, then when the SNR is taken to infinity by increasing the power constraint $P$ while holding the noise spectrum fixed: \beq{D_high_prop}
\lim_{\tSNR\rightarrow\infty} \frac{\SDR^{opt}}{\tSNR^\rho}
 = \Gamma_S \Gamma_C \ \ . \eeq
\end{prop}
\vspace{5mm}

\subsection{Predictive Presentation of the Gaussian RDF and Capacity}
\label{Sub_DPCM}

\begin{figure}
\centering
      \subfloat[RDF-achieving configuration using filters and AWGN channel.]{\label {Fig_Pred_RDF}
      \input{pred_rdf.pstex_t}}
      \\
      \subfloat[Capacity-achieving configuration using filters and white Gaussian input.]{\label {Fig_Pred_C}
      \input{pred_c.pstex_t}}
      \ccaption{Forward-channel configurations for the RDF and capacity.} \label{Figs_Pred}
\end{figure}

Not only the SLB and SUB in \eqref{SLB} and \eqref{SUB} can be
written in predictive forms, but also the rate-distortion function
and channel capacity, in the Gaussian case. These predictive forms
are given in terms of the forward-channel configurations depicted in
\figref{Figs_Pred}.

For source coding, let $F_1(\ej)$ be any filter with amplitude
response satisfying \beq{source_pre} |F_1(\ej)|^2 = 1 -
\frac{D(\ej)}{S_S(\ej)} \ \ , \eeq where $D(\ej)$ is the distortion
spectrum materializing the water-filling solution \eqref{RDF}. We
call $F_1(\ej)$ and $F_2(\ej)=F_1^*(\ej)$ the pre- and post-filters
for the source $S$ \cite{FederZamir94}.

As a consequence of
\eqref{noisy_prediction}, the pre/post filtered AWGN depicted in
\figref{Fig_Pred_RDF} satisfies \cite{ZamirKochmanErezDPCM-IT}: \beq{Pred_RDF} R(D) =
\frac{1}{2}\log\left(1+\frac{\var\{U_n|V_{-\infty}^{n-1}\}}{\var\{Z_n\}}
\right) \ \ ,
\eeq where ${\var\{Z_n\}}=\theta_S$.
Note that in the limit of low
distortion the filters vanish, prediction from $U_n$ is equivalent
to prediction from $V_n$, and we go back to \eqref{SLB}. Defining
the source Wiener coefficient \beq{alpha_S} \alpha_S = 1 - \exp
\left( -2R(D)\right) \ \ , \eeq \eqref{Pred_RDF} implies that
\beq{source_exact} \var\{U_n|V_{-\infty}^{n-1}\} =
\frac{\alpha_S}{1-\alpha_S} \theta_S \ \ . \eeq

For channel coding, let $G_1(\ej)$ be any filter with amplitude
response satisfying \beq{channel_pre} |G_1(\ej)|^2 =
\frac{P(\ej)}{\theta_C} \ \ , \eeq where $P(\ej)$ and $\theta_C$ are
the channel input spectrum and water level materializing the
water-filling solution \eqref{C}. $G_1(\ej)$ is usually referred to
as the channel shaping filter, but motivated by the the similarity
with the solution to the source problem we call it a channel
pre-filter. At the channel output we place $G_2(\ej)=G_1^*(\ej)$,
known as a matched filter, which we call a channel post-filter.

In the pre/post
filtered colored-noise channel depicted in \figref{Fig_Pred_C}, let
the input $\tilde X_n$ be white and define the (non-Gauusian, non-additive) noise $\tilde Z_n= \tilde Y_n -
\tilde X_n$. Then the channel satisfies (see \cite{Forneyallerton04}, \cite{ZamirKochmanErezDPCM-IT}):
\beq{Pred_C} C =
\frac{1}{2}\log\left(\frac{\var\{\tilde X_n\}}{\var\{\tilde Z_n|
{\tilde Z}_{-\infty}^{n-1}\}} \right)
\eeq where $\var\{\tilde X_n\} = \theta_C$. Note that in the limit of low noise the
filters vanish, prediction from $\tilde Z_n$ is equivalent to
prediction from $Z_n$, and we go back to \eqref{SUB}. Defining the
channel Wiener coefficient \beq{alpha_C} \alpha_C = 1 - \exp
\left( -2C\right) \ \ , \eeq \eqref{Pred_C} implies that
\beq{channel_exact} \var\{\tilde Z_n|{\tilde Z}_{-\infty}^{n-1}\} =
\frac{1-\alpha_C}{\alpha_C} \theta_C \ \ . \eeq

Combining \eqref{alpha_S} with \eqref{alpha_C}, we note that in a joint source-channel setting where the optimum performance \eqref{OPTA1} is achieved, \beq{alpha} \alpha_S=\alpha_C=\alpha \ \ . \eeq A connection between the water-filling parameters and conditional variances can be derived using \eqref{source_exact} and \eqref{channel_exact}.

The predictive presentations \eqref{Pred_RDF} and \eqref{Pred_C} translate the process mutual information rates $\bar I(S_n;\hat S_n)$ and $\bar I(X_n;Y_n)$ to the conditional mutual informations $I(U_n;V_n|V_{-\infty}^{n-1})$ and $I(\tilde X_n;\tilde X_n + \tilde Z_n|\tilde Z_{-\infty}^{n-1})$, respectively.  This is highly attractive as the
basis for coding schemes, since it allows to use the combination of predictors and 
generic optimal codebooks for \emph{white} sources and channels,
regardless of the actual spectra, without compromising optimality.
See e.g. \cite{GuessVaranasiIT}, \cite{ZamirKochmanErezDPCM-IT}. In the source case, \eqref{Pred_RDF} 
establishes the optimality of a DPCM-like scheme, where the prediction error of
$U_n$ from the past samples of $V_n$ is being quantized and the
quantizer is equivalent to an AWGN. For channel coding, \eqref{Pred_C} implies a noise-prediction receiver, 
which can be shown to be
equivalent to the better known MMSE FFE-DFE solution
\cite{CDEF-MMSE-DFE}.

\subsection{Good Lattices for Quantization and Channel Coding} \label{Sub_lattice}

Let $\Lambda$ be a $K$-dimensional lattice, defined by the generator
matrix $G \in \mathbb{R}^{K \times K}$. The lattice includes all
points $\{\bl=G\cdot\bi:\bi\in\mathbb{Z}^K\}$ where
$\mathbb{Z}=\{0,\pm 1, \pm 2, \ldots\}$. The nearest neighbor
quantizer associated with $\Lambda$ is defined by \[Q(\bx) =
\arg\min_{\bl\in\Lambda} \|\bx-\bl\| \ \ ,\] where ties are broken in a systematic way. 
Let the basic Voronoi
cell of $\Lambda$ be
\[\Nu_0 = \{\bx : Q(\bx)=\textbf{0}\}\ \ .\] The second
moment of a lattice is given by the variance of a
uniform distribution over the basic Voronoi cell, per dimension:
\beq{lattice_power} \sigma^2(\Lambda) = \frac{1}{K} \cdot
\frac{\int_{\Nu_0} \|\bx\|^2 d\bx}{\int_{\Nu_0} d\bx} \ \ . \eeq The
modulo-lattice operation is defined by:
\[\bx \bmod{\Lambda} = \bx - Q(\bx)\ \ .\] The following is the key condition that has 
to be verified in the analysis to follow. \begin{definition}\label{correct_decoding} (\textbf{Correct decoding}) We say that
correct decoding of a vector $\bx$ by a lattice $\Lambda$
occurs, whenever \[\bx \bmod \Lambda = \bx \ \ , \] i.e., $\bx \in \Nu_0$. \end{definition} 
For a dither vector $\bd$ which is independent of $\bx$ and
uniformly distributed over the basic Voronoi cell $\Nu_0$,
$[\bx+\bd] \bmod{\Lambda}$ is uniformly distributed over $\Nu_0$ as
well, and is independent of $\bx$ \cite{FederZamirLQN}.

We assume the use of lattices which are simultaneously good for
source coding (MSE quantization) and for AWGN channel coding
\cite{GoodLattices}. Roughly speaking, a sequence of $K$-dimensional
lattices is \emph{good for MSE quantization} if the second moment of
these lattices tends to that of a uniform distribution over a ball of the same volume, as $K$
grows. A sequence of lattices is \emph{good for AWGN channel coding}
if the probability of correct decoding \eqref{correct_decoding} of a
Gaussian i.i.d. vector with element variance smaller than the
square radius of a ball having the same volume as the lattice basic cell,
approaches zero for large $K$. There exists a sequence of lattices
satisfying both properties simultaneously, thus for these lattices,
correct decoding holds with high probability for Gaussian i.i.d.
vectors with element variance smaller than $\sigma^2(\Lambda)$, for
large enough $K$. This property also holds when the Gaussian vector
is replaced by a linear combination of Gaussian and ``self noise''
(uniformly distributed over the lattice basic cell) components. The following formally states the property used in the sequel.

\begin{definition} \label{def_Combination} (\textbf{Combination noise})  Let $\bZ_1,\ldots,\bZ_L$ be mutually-i.i.d. vectors, independent of $\bZ_0$, uniformly distributed over the basic cell of $\Lambda$, and let $\bZ_0$ be a Gaussian i.i.d. vector with element variance $\sigma^2(\Lambda)$. Then for any real coefficients, $\sum_{l=0}^L \alpha_l \bz_L$ is a combination noise with composition $\alpha_0,\ldots,\alpha_L$. \end{definition}

\vspace{5mm}
\begin{prop} \label{prop_lattice} (\textbf{Existence of good lattices})
Let $\{\Lambda_K\}$ denote a sequence of $K$-dimensional lattices with basic cells $\{\Nu_K\}$ of fixed
second moment $\sigma^2$. Let $\{\bZ_K\}$ be a corresponding sequence of combination noise vectors, with fixed composition satisfying:
\[ \sum_{l=1}^L \alpha_l^2 <1 \ \ . \] Then there exists a sequence $\{\Lambda_K\}$ such that: \[\limsup_{K\rightarrow\infty} \Pr\{ \bZ_K \bmod \Lambda_K \neq \bZ_k \} = 0 \ \ . \] \end{prop}  

This is similar to \cite[Poposition 1]{JointWZ-WDP}, but with the single ``self-noise'' component replaced by a combination. It therefore requires a proof, included in the appendix. 

\subsection{Coding for the Joint WZ/DPC Problem using Modulo-Lattice
Modulation} \label{Sub_joint}

\begin{figure}[t]
\centering
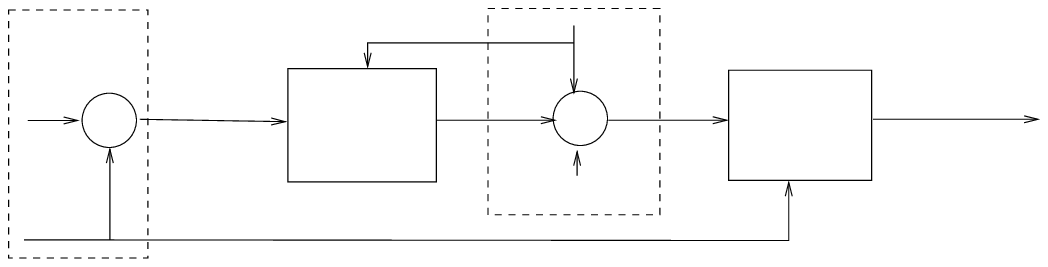
\ccaption{The Wyner-Ziv / dirty-paper coding problem.}
\label{joint_prob_fig}
\end{figure}

\begin{figure}[!t]
\vspace{1cm}
\centering \subfloat[The MLM scheme.]{\label{joint_Scheme_fig}
\input{joint_scheme3.pstex_t}}
\\ \vspace{2mm}
      \subfloat[Asymptotic equivalent real-additive channel for ``good'' lattices.
      The equivalent noise
	$\bZ_{eq} = \alpha \bZ - (1-\alpha) \bX$ is independent of $\bQ$, and asymptotically Gaussian.]
	{\label {output_eq_fig}
      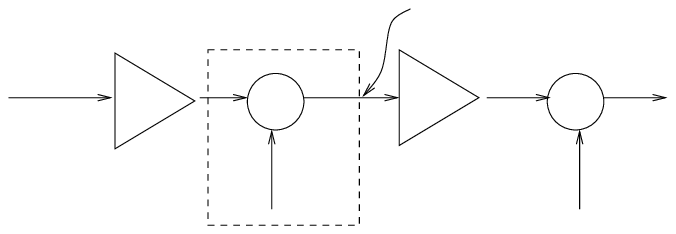}
      \ccaption{MLM Wyner-Ziv / dirty-paper coding.}
\label{joint_equivalent_fig}
\end{figure}

The lattices discussed above can be used for achieving the optimum
performance in the joint source/channel Gaussian Wyner-Ziv/dirty-paper coding, \footnote{An alternative form of this scheme may be
obtained by replacing the lattice with a random code and using
mutual information considerations; see
\cite{WilsonNarayananCaire07}.} depicted in
\figref{joint_prob_fig}. In that problem, the source is the sum of
an unknown i.i.d. Gaussian component $Q_n$ and an arbitrary
component $J_n$ known at the decoder, while the channel noise is the
sum of an unknown i.i.d. Gaussian component $Z_n$ and an arbitrary
component $I_n$ known at the encoder. In \cite{JointWZ-WDP} the MLM
scheme of \figref{joint_Scheme_fig} is shown to achieve the optimal performance \eqref{OPTA1} for
suitable $\alpha$ and $\beta$. This is done showing asymptotic
equivalence with high probability (for good lattices) to the real-additive channel of
\figref{output_eq_fig}. The \emph{output}-power constraint $P$ in that last
channel reflects the element variance condition in order to ensure
correct decoding of the vector $\beta \bQ_n +
{\bZ_{eq}}_n$ with high probability, according to Proposition~\ref{prop_lattice}. When this holds, the dithered
modulo-lattice operation at the encoder and the decoder perfectly
cancel each other. This way, the MLM scheme asymptotically
translates the SI problem to the simple problem of transmitting the
unknown source component $Q_n$ over an AWGN, where the known source component $J_n$ and the channel interference
$I_n$ are not present.


\section{The AM Scheme} \label{Sec_Matching}

\begin{figure*}[t]
\centering
\input{main_scheme4.pstex_t}
\ccaption{The Analog Matching scheme.} \label{matching_fig}
\end{figure*}

In this section we prove the optimality of the Analog
Matching scheme, depicted in \figref{matching_fig}, in the limit of high lattice
dimension. We assume for now that we have $K$ mutually-independent identically-distributed
source-channel pairs in parallel,\footnote{We will discuss in the
sequel how this leads to optimality for a single source and a single
channel.} which allows a $K$-dimensional dithered modulo-lattice
operation across these pairs. Other operations are done
independently in parallel. To simplify notation we omit the
index $k$ of the source/channel pair ($k=1,2,\ldots,K$), and use
scalar notation meaning \emph{any} of the $K$ pairs; we denote
by bold letters $K$-dimensional vectors, for the
modulo-lattice operation. Subscripts denote time instants. Under
this notation, the AM encoder is given by:
\beqn{AM_encoder} \nk{U} &=& {f_1}_n * \nk{S} \nonumber \\
\tilde\bX_n &=& \left[\beta\bU_n - \bI_n + \bD_n \right] \bmod
\Lambda \nonumber \\ \nk{I} &=& - \sum_{m=1}^\infty {p_C}_m \tilde
X_{n-m} \nonumber
\\ \nk{X} &=& {g_1}_n * \nk{\tilde X}   \ \ , \eeqn while the decoder is
given by: \beqn{AM_decoder} \nk{\tilde Y} &=& {g_2}_n * \nk{Y} \nonumber \\
\nk{Y'} &=& \nk{\tilde Y} - \sum_{m=1}^\infty
{p_C}_m \tilde Y_{n-m} \nonumber \\
\bV_n &=& \frac{1}{\beta} \left[\tilde \bY_n - \beta \bJ_n - \bD_n
\right] \bmod \Lambda + \bJ_n \nonumber \\ \nk{J} &=&
\sum_{m=1}^\infty {p_S}_k V_{n-k}  \nonumber
\\ \nk{\hat S} &=& {f_2}_n * \nk{V}   \ \ , \eeqn where $*$ denotes
convolution. For each filter of frequency response $H(\ej)$, the corresponding impulse response is denoted by small letters $h_n$. Each of the $K$
parallel channels is given by the colored noise
model \eqref{channel}. 

The filters used in the scheme are determined by the optimal solutions presented in Sections~\ref{Sub_SLB} and \ref{Sub_DPCM}. The channel capacity, and corresponding water-level $\theta_C$, are given by \eqref{C}. This determines, through the optimality condition \eqref{OPTA1}, the distortion level $D$. Using that $D$, the RDF, and corresponding water-level $\theta_S$, are given by \eqref{RDF}. 
The filters $F_1(\ej)$ and $F_2(\ej)$ are then chosen according to \eqref{source_pre}, and
 $G_1(\ej)$ and $G_2(\ej)$ according to \eqref{channel_pre}. We also use $\alpha$ of \eqref{alpha}. Finally, $P_S(\ej)$ and $P_C(\ej)$ are the optimal predictors \eqref{optimal_predictor} of the spectra \beq{S_V} S_V(\ej) \Ddef |F_1(\ej)|^2 S_S(\ej) + \frac{1-\alpha}{\beta^2} \theta_C \eeq and \beq{S_tildeZ} S_{\tilde Z}(\ej) \Ddef \ \Bigl(1-|G_1(\ej)|^2\Bigr)^2 \theta_C +
|G_1(\ej)|^2 S_Z(\ej) \ \,  \eeq respectively, where we take $|G_1(\ej)|^2 S_Z(\ej)=0$ wherever $|G_1(\ej)|=0$ even if $S_Z(\ej)$ is infinite.

The analysis we apply to the scheme shows that at each time instant it is
equivalent to a joint source/channel side-information (SI) scheme,
and then applies the Modulo-Lattice Modulation (MLM) approach presented in \secref{Sub_joint}. The key to the proof
is showing that with high probability the correct decoding event of Definition~\ref{correct_decoding}
holds, thus the modulo-lattice operation at the decoder exactly
cancels the corresponding operation at the encoder. As the distribution of the signal fed to the decoder modulo-lattice operation depends upon the past decisions through the filters memory, the analysis has a \emph{recursive} nature: we show, that if the scheme is in a ``correct state'' at time instant $n$, it will stay in that state at instant $n+1$ with high probability, resulting in optimal distortion. Formally, for the decoder modulo-lattice input:
\beq{T_n} \bT_n 
 = \beta (\bU_n - \bJ_n) + {\bZ_{eq}}_n \ \ , \eeq we define the desired state as follows. 

\begin{definition} \label{def_init} We say that the Analog Matching
scheme is \emph{correctly initialized} at time instance $n$, if all
signals at all times $n-1,n-2,\ldots$ take values according to the
assumption that correct decoding (see Definition~\ref{correct_decoding}) held w.r.t. $\bT_{n-1},\bT_{n-2},\ldots$.\end{definition}

Using this, we make the following optimality claim.

\vspace{5mm}
\begin{theorem} \label{thm_AM2} (\textbf{Asymptotic optimality of the Analog Matching
scheme}) Let $D(N,K)$ be the achievable expected distortion of the AM scheme operating on input blocks of time duration $N$ with lattice dimension $K$.
Then there exists a sequence $N(K)$ such that
\[\lim_{K\rightarrow\infty} D(N(K),K) = D^{opt} \ \ , \]
where $D^{opt}$ was defined
in \eqref{OPTA1}, provided that 
at the start of transmission the scheme is correctly initialized.
\vspace{3mm}
\end{theorem}

In \secref{sub_high} we gain some insight into the workings of the scheme by considering it in the uqual-BW high-SNR regime, while in \secref{sub_BW} we cosider the important special cases of bandwidth expansion and compression. \secref{sub_optimality} contains the proof of \thref{thm_AM2}, and then in \secref{sub_remarks} we discuss how one can implement a scheme based on the theorem.

\subsection{The Scheme in the Equal-BW High-SNR Regime} \label{sub_high}

In the equal-BW case ($B_S=B_C$) and in the limit of high resolution ($D^{opt} \ll P_e(S_S)$), the scheme can be simplified significantly. In this limit, the filters approach \emph{zero-forcing} ones: both source and channel pre- and post-filters collapse to unit all-pass filters, while the source and channel predictors become just the optimal predictors of the spectra $S_S(\ej)$ and $S_Z(\ej)$, respectively. The receiver filter $1-P_C(z)$ is then a whitening filter for the noise, and the channel from $X_n$ to $Y'_n$ is equivalent to a white-noise inter-symbol interference (ISI) channel: \beq{ISI} Y'_n = X_n - \sum_{m=1}^\infty {p_C}_m X_{n-m} + W_n \ \ , \eeq where $W_n$ is AWGN of variance $P_e(S_Z)$. Under these conditions, the source can always be modeled as an auto-regressive (AR) process: \beq{AR} S_n =  \sum_{m=1}^\infty {p_S}_m S_{n-m} + Q_n \ \ , \eeq where $Q_n$ is the white innovations process, of power $P_e(S_S)$. 

\begin{figure*}[!t]
\centering
\input{simple_scheme6.pstex_t}
\ccaption{The scheme in the high-SNR limit.}
\label{fig_high_resolution}
\end{figure*}

The resulting scheme is depicted in \figref{fig_high_resolution}. It is evident, that the channel predictor cancels all of the ISI, while the source predictor removes the source memory, so that effectively the scheme transmits the source innovation through an AWGN channel. The gains of source and channel prediction are $\Gamma_S = P_e(S_S) / \var\{S_n\}$ and $\Gamma_C = P_e(S_Z) / \var\{Z_n\}$, respectively (recall \eqref{pred_gain}). In light of \eqref{D_high_prop}, the product of the two is indeed the required gain over memoryless transmission. In fact, if we assume that the modulo-lattice operations have no effect, then the entire scheme is equivalent to the AWGN channel: \[\hat S_n = S_n + \frac{W_n}{\beta} \ \ . \] Letting $\beta^2 = P / \var\{Q_n\}$, we than have that: \[ \SDR = \frac{\beta^2 \var\{S_n\}}{\var\{W_n\}} = \frac{\beta^2\Gamma_C\SNR\var\{S_n\}}{P} = \Gamma_S\Gamma_C\SNR \] which is optimal at the limit, recall \eqref{D_high_prop}. In order to satisfy the power constraint, the lattice second moment must be $P$, thus the gain $\beta$ amplifies the source innovations to a power equal to the lattice second moment; as we will prove in the sequel, this choice of $\beta$ indeed guarantees correct decoding, on account of Proposition~\ref{prop_lattice}.

\subsection{The BW Mismatch Case}\label{sub_BW}

At this point, we present the special cases of \emph{bandwidth
expansion} and \emph{bandwidth compression}, and see how the analog
matching scheme specializes to these cases. In these cases the
source and the channel are both white, but with different bandwidth
(BW). The source and channel prediction gains are both one, and the
optimum condition \eqref{OPTA} becomes: \beq{OPTA_white} \SDR^{opt}
= \left(1+\tSNR\right)^\rho \ \ , \eeq where the bandwidth ratio
$\rho$ was defined in \eqref{rho}.

For bandwidth expansion ($\rho>1$), we choose to work with a sampling rate
corresponding with the channel bandwidth, thus in our discrete-time
model the channel is white, but the source is band-limited to a
frequency of $\frac{1}{2\rho}$. As a result, the channel predictor
$P_C(z)$ vanishes and the channel post-filters become the
scalar Wiener factor $\alpha$. The source water-filling solution
allocates all the distortion to the in-band frequencies, thus we
have $\theta_S = \rho D$ and the source pre- and post-filters become
both ideal low-pass filters of width $\frac{1}{2\rho}$ and height
\beq{filter_height} \sqrt{1-\frac{1}{\SDR^{opt}}}=\sqrt{1-\frac{1}{\left(
1+\tSNR\right) ^{\rho}}} \ \ . \eeq As the source is band-limited,
the source predictor is non-trivial and depends on the distortion
level. The resulting prediction error of $U_n$ has variance
\[ \var\{U_n|V_{-\infty}^{n-1}\} = \frac{\rho \var\{S_n\}}
{\left( 1+\tSNR\right) ^{\rho-1}} \ \ , \] and the resulting
distortion achieves the optimum \eqref{OPTA_white}.

For bandwidth compression ($\rho<1$), the sampling rate reflects the source
bandwidth, thus the source is white but the channel is band-limited
to a frequency of $B_C=\frac{\rho}{2}$. In this case the source
predictor becomes redundant, and the pre- and post-filters become a
constant factor equal to \eqref{filter_height}. The channel pre- and
post-filters are ideal low-pass filter of width $\frac{\rho}{2}$ and
unit height. The channel predictor is the SNR-dependent DFE. Again
this results in achieving the optimum distortion \eqref{OPTA_white}.
It is interesting to note, that in this case the outband part of the
channel error $\nk{\tilde Z}$ is entirely ISI (a filtered version of
the channel inputs), while the inband part is composed of both
channel noise and ISI, and tends to be all channel noise at high
SNR.

\subsection{Proof of \thref{thm_AM2}} \label{sub_optimality}

We start the optimality proof by showing that the Analog Matching scheme is equivalent at
each time instant to a WZ/DPC scheme, as in \secref{Sub_joint}. Specifically, the equivalent scheme is shown in
\figref{side_info_eq_fig}, which bears close resemblance to
\figref{joint_Scheme_fig}. The equivalence is immediate using the definitions of $I_n$ \eqref{AM_encoder} and $J_n$ \eqref{AM_decoder}, since they are constructed in
the encoder and the decoder using \emph{past} values of $\tilde X_n$
and $V_n$, respectively, thus at any fixed time instant they can be seen
as side information. It remains to show that indeed the unknown noise component is white, and
evaluate its variance.

\vspace{5mm}
\begin{lemma} \label{lemma_AM1} \textbf{(Equivalent side-information scheme)}
Assume that $\var{\tilde X_n} = \theta_C$, then
\[\nk{Z'} \Ddef \frac{Y'_n-I_n}{\alpha} - \tilde X_n \] is a white process, independent of all $U_n$,
with variance
\[\var \{Z'_n\} = \frac{1-\alpha}{\alpha} \theta_C \ \ .
\] \vspace{3mm}
\end{lemma}

\begin{figure}[t]
    \centering
    \subfloat[Equivalent WZ/DPC scheme.]{\label
    {side_info_eq_fig}
      \input{AM_as_joint.pstex_t}}
      \\
    \subfloat[Equivalent modulo-lattice channel.]{\label
    {finite_k_matching_fig}
      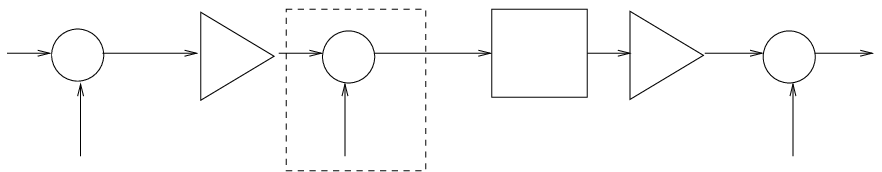}
      \\
      \subfloat[Asymptotic equivalent scalar additive channel for good lattices.]{\label {output_matching_fig}
      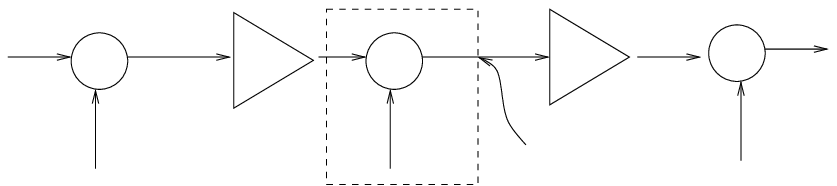}
      \ccaption{Equivalent channels for the Analog Matching scheme.}
\label{matching_equivalent_fig}
\end{figure}

The proof of this lemma appears in the appendix. Now we note, that if the modulo-lattice operations in the equivalent scheme of \figref{side_info_eq_fig} can be dropped, this will result in a scalar additive white-noise channel (see \figref{matching_equivalent_fig}):  \beq{ideal_additive} V_n = U_n +
\frac{{Z_{eq}}_n}{\beta} \ \ , \eeq where \beq{Z_eq} {Z_{eq}}_n = \alpha Z'_n - (1-\alpha) \tilde X_n \eeq is a white additive noise term of variance \beq{var_Z_eq} \var\{{Z_{eq}}_n\} = \alpha \var \{Z'_n\} = (1-\alpha) \theta_C \ \ . \eeq Together with the source pre/post filters $F_1(\ej)$ and $F_2(\ej)$, we can have the forward test channel of \figref {Fig_Pred_RDF}. Furthermore, if $\beta^2$ equals \beq{beta_0} \beta_0^2 \Ddef (1-\alpha)
\frac{\theta_C}{\theta_S} \eeq then the additive noise variance is $\theta_S$, resulting in optimal performance. The relevant condition is that correct decoding (recall Definition~\ref{correct_decoding}) holds for $\bT_n$ \eqref{T_n}. Using the concept of correct initialization (Definition~\ref{def_init}), we first give a recursive claim.

\begin{lemma} \label{lemma_AM}  (\textbf{Steady-state behavior of the Analog Matching
scheme}) Assume that the Analog Matching scheme is applied,
using a lattice $\Lambda=\Lambda_K$ of dimension $K$ which
 is taken from a sequence of lattices of second moment $\theta_C$ which are
simultaneously good for source and channel coding in the sense of
Proposition~\ref{prop_lattice}. Then the probability that correct decoding
does not hold in the present instance can be bounded by $p_e(K)$, where  \[\lim_{K\rightarrow\infty} p_e(K) = 0 \ \ , \] \vspace{3mm}
given that the scheme is correctly initialized and that $\beta>\beta_0$ \eqref{beta_0}.  \vspace{3mm} \end{lemma}

Now we translate the conditional result above (again proven in the appendix), to the optimality
claim for blocks. 

\textbf{Proof of \thref{thm_AM2}:} We choose $N(K)$ to be some sequence such that \[\lim_{K\rightarrow\infty}N(K)=\infty \ \ , \] but at the same time  \[\lim_{K\rightarrow\infty}N(K)p_e(K)=0 \ \ , \] where $p_e(K)$ was defined in Lemma~\ref{lemma_AM}. Let $D^{correct}$ and $D^{incorrect}$ be the expected distortion given that the scheme remains correctly initialized at the end of transmission or does not, respectively. By the union bound we have that: \[D(N(K),K) \leq D^{correct}(N(K),K) + N(K)p_e(K) D^{incorrect}(N(K),K) \ \ . \] Since we assumed that $N(K)p_e(K)$ vanishes in the limit of infinite $K$, so does the second term; see \cite[Appendix II-B]{JointWZ-WDP}. We thus have that \[ \lim_{K\rightarrow\infty} D(N(K),K) = \lim_{K\rightarrow\infty} D^{correct}(N(K),K) \] and we can assume that \eqref{ideal_additive} holds throughout the block. This results in the forward channel of \figref{Fig_Pred_RDF}, up to two issues. First, the channel is stationary while transmission has finite duration, and second the additive noise variance is larger than $\theta_S$ since $\beta>\beta_0$. The first may be solved by forcing $U_n$ and $V_n$ to be zero outside the transmission block, resulting in an excess distortion term; however this finite term vanishes when averaging over large $N(K)$. The second implies that $D^{opt}+\epsilon$ may still be achieved for any $\epsilon>0$, and the result follows by a standard arguments, replacing $\epsilon$ by a sequence $\epsilon(K)\rightarrow 0$

\subsection{From the Idealized Scheme to Implementation} \label{sub_remarks}

We now discuss how the scheme can be implemented with finite
filters, how the correct initialization assumption may be dropped,
and how the scheme may be used for a single source/channel pair.

1. \textbf{Filter length.} If we constrain the filters to have
finite length, we may not be able to implement the optimum filters.
However, it is possible to show, that the effect on both the correct
decoding condition and the final distortion can be made as small
as desired, since the additional signal errors due to the filters
truncation can be all made to have arbitrarily small variance by taking
long enough filters. In the sequel we assume the filters all have
length $L$.

2. \textbf{Initialization.} After taking finite-length filters, we
note that correct initialization now only involves a \emph{finite}
history of the scheme. Consequently, we can create this state by
adding a finite number of channel uses. Now we may create a valid
state for the channel predictor $P_C(\ej)$ by transmitting $L$
values $\tilde X_n = 0$; see \cite{GuessVaranasiIT}. For the source
predictor the situation is more involved, since in absence of past
values of $\nk{V}$, the decoder cannot reduce the source power to
the innovations power, and correct decoding may not hold. This can
be solved by de-activating the predictor for the first $L$ values of
$U_n$, and transmitting them with lower $\beta$ such that
\eqref{AM_Condition1} holds without subtracting $J_n$. Now in order
to achieve the desired estimation error for these first values of
$V_n$, one simply repeats the same values of $U_n$ a number of times
according to the (finite) ratio of $\beta$'s. If the block length
$N$ is long enough relative to $L$, the number of excess channel
uses becomes insignificant.

3. \textbf{Single source/channel pair.} A pair of
interleaver/de-interleaver can serve to emulate $K$ parallel
sources, as done in \cite{GuessVaranasiIT} for an FFE-DFE receiver,
and extended to lattice operations in \cite{RamiShamaiUriLattices}. 
Interestingly, while a separation-based scheme which employs time-domain processing 
for both source and channel parts requires two separate interleavers, one suffices for the AM scheme. Together with the initialization process, we have the following algorithm.

\emph{encoder:}

1. Write $U_n$ row-wise into an interleaving table.

2. Before each row, add source initialization samples. 

3. Build a table for $\tilde X_n$, starting by zero columns for channel initialization. Then add more columns using column-wise modulo-lattice operations on the table of $U_n$, using row-wise past values of $\tilde X_n$ as inputs to the channel predictor.

4. Feed the $\tilde X_n$ table to the channel pre-filter row-wise. 

\emph{decoder:}

1. Write $Y'_n$ row-wise. 

2. Discard the first columns, corresponding to channel initialization.

3. Build a table for $V_n$, starting by using the source initialization data. Then add more columns using column-wise modulo-lattice operations on the table of $Y'_n$, using row-wise past values of $V_n$ as inputs to the source predictor.

4. Feed the $V_n$ table to the source post-filter row-wise.

\section{Unknown SNR} \label{Sec_Universal}

So far we have assumed in our analysis that both the encoder and
decoder know the source and channel statistics. In many practical
communications scenarios, however, the encoder does not know the
channel, or equivalently, it needs to send the same message to
different users having different channels.  Sometimes it is assumed
that the channel filter $H_0(\ej)$ is given, but the noise level $N$
is only known to satisfy $N\leq N_0$ for some given $N_0$. For this
special case, and specifically the broadcast bandwidth expansion and
compression problems, see
\cite{ShamaiVerduZamir,MittalPhamdo,TzvikaBroadcast,Puri2005}.

Throughout this section, we demonstrate that the key factor in
asymptotic behavior for high SNR is the bandwidth ratio $\rho$
\eqref{rho}. We start in \secref{Sub_Universal_Theorem} by proving a
basic lemma regarding achievable performance when the encoder is not
optimal for the actual channel. In the rest of the section we
utilize this result: in \secref{Sub_Optimal} we show asymptotic
optimality for unknown SNR in the case $\rho=1$, then in
\secref{Sub_BW_Change} we show achievable performance for the
special cases of (white) BW expansion and compression, and finally
in \secref{Sub_High} we discuss general spectra in the high-SNR
limit.

\subsection{Basic Lemma for Unknown SNR} \label{Sub_Universal_Theorem}

We prove a result which is valid for the transmission of a colored
source over a degraded colored Gaussian broadcast channel: We assume
that the channel is given by \eqref{channel}, where $B_C$ is known
but the noise spectrum $S_Z(\ej)$ is unknown, except that it is
bounded from above by some spectrum $S_{Z0}(\ej)$ everywhere. We
then use an Analog Matching encoder optimal for $S_{Z0}(\ej)$, as in
\thref{thm_AM2}, but optimize the decoder for the actual noise
spectrum. Correct decoding under $S_{Z0}(\ej)$ ensures correct
decoding under $S_Z(\ej)$, thus the problem reduces to a
\emph{linear} estimation problem, as will be evident in the proof.

For this worst channel $S_{Z0}(\ej)$ and for optimal distortion
\eqref{OPTA}, we find the water-filling solutions
\eqref{RDF},\eqref{C}, resulting in the source and channel water
levels $\theta_S$ and $\theta_C$ respectively, and in a
\emph{source-channel passband} $\cF_0$, which is the intersection of
the inband frequencies of the source and channel water-filling
solutions: \beqn{sourcechannelpassband} \cF_S & = & \{f:
S_S(\ej)\geq \theta_S \} \ \ ,  \nonumber \\ \cF_C & = & \{f:
S_{Z_0}(\ej)\leq \theta_C \} \ \ , \nonumber \\ \cF_0 & = & \cF_S
\cap \cF_C \ \ . \eeqn Under this notation we have the following lemma, proven in the appendix. It shows that the resulting distortion spectrum is that of a linear scheme which transmits the source into a channel with noise spectrum $P / \Phi(\ej)$, where \beq{in_thm} \Phi(\ej) =
\frac{S_{Z_0}(\ej)}{S_Z(\ej)} \left[1-\frac{ S_{Z_0}(\ej) -
S_Z(\ej)}{\theta_C}\right]\frac{S_S(\ej)-\theta_S}{\theta_S} \ \  \nonumber \eeq
depends on both the design noise spectrum and the actual one.

\vspace{5mm} \begin{lemma} \label{lemma_universal} For any noise
spectrum $S_{Z0}(\ej)$, there exists an encoder, such that for any
equivalent noise spectrum \beq{degraded} S_Z(\ej)\leq S_{Z0}(\ej) \
\forall \ f\in\cF_C \ \ , \eeq a suitable decoder can arbitrarily
approach:
\[D = \int_{-\frac{1}{2}}^{\frac{1}{2}} D(\ej) df \ \
, \vspace{-2mm} \] where the distortion spectrum $D(\ej)$ satisfies:
\beq{universal_dist_rep} {D(\ej)} = \left\{
\begin{array}{ll}
\frac{ S_S(\ej)}{1 + \Phi(\ej)} , & \mbox{if $f\in\cF_0$} \\
\min\Bigl( S_S(\ej),\theta_S\Bigr) , & \mbox{otherwise}
\end{array}
\right\} \ \ . \eeq
\vspace{3mm}
\end{lemma}

\textbf{Remarks:}

1. Outside the source-channel passband $\cF_0$, there is no gain
when the noise spectrum density is lower than expected. Inside
$\cF_0$, the distortion spectrum is strictly monotonously decreasing
in $S_Z(\ej)$, but the dependence is never stronger than inversely
proportional. It follows, that the overall SDR is at most linear
with the SNR. This is to be expected, since all the gain comes from
linear estimation.

2. In the unmatched case modulation may change performance. That is,
swapping source frequency bands before the analog matching encoder
will change $\cF_0$ and $\Phi(\ej)$, resulting in different
performance as $S_Z(\ej)$ varies. It can be shown that the best
robustness is achieved when $S_S(\ej)$ is monotonously decreasing in
$S_Z(\ej)$.

3. The degraded channel condition \eqref{degraded} is not necessary.
A tighter condition for correct decoding to hold can be stated in
terms of $S_S(\ej)$, $S_{Z0}(\ej)$ and $S_Z(\ej)$: the integral over $S_{eq}(\ej)$, defined in the appendix \eqref{S_eq_robust}, must be at most as it is for the spectrum $S_0(\ej)$.

\subsection{Asymptotic Optimality for equal BW} \label{Sub_Optimal}

We prove asymptotic optimality in the sense that, if in an ISI channel (recall \eqref{ISI}), the ISI filter is known but the SNR is
only known to be above some $\SNR_0$, then a single encoder can
simultaneously approach optimality for any such SNR, in the limit
of high $SNR_0$. 

\vspace{5mm} \begin{theorem} \label{thm_equal_BW}
\textbf({High-SNR robustness}) Let the source and channel have BW $B_S=B_C=1$, and let the equivalent ISI model of the channel
\eqref{ISI} have fixed filter coefficients (but unknown innovations power $\var\{W_n\}$). Then, there exists
an \emph{SNR-independent} sequence of encoders indexed by their lattice dimension $K$, each achieving $\SDR_K(\SNR)$, such that for any $\delta>0$:
 \[ \lim_{K\rightarrow\infty} \SDR_K (\SNR) \geq (1-\delta)\SDR^{opt}(\SNR)
\]
for sufficiently large (but finite) SNR, i.e., for all
$\SNR\geq\SNR_0(\delta)$. \vspace{3mm} \end{theorem}

\begin{proof} The limit of a sequence of encoders is required, since any fixed finite-dimensional encoder has a gap from $\SDR^{opt}$ that would limit performance as $\SNR\rightarrow\infty$. At the limit, however, we may assume an ideal scheme.
In terms of the colored noise channel \eqref{channel}, the unknown noise variance in the theorem conditions is equivalent to having noise spectrum \[S_Z(\ej) = \frac{\SNR_0}{\SNR}S_{Z0}(\ej) \] where $\SNR \geq \SNR_0 = \SNR_0(\delta)$.  
We apply Lemma~\ref{lemma_universal}, with an encoder designed for $S_{Z0}(\ej)$.
If the source spectrum is bounded away from zero and the
$S_{Z0}(\ej)$ is bounded from above, we can always take $\tSNR_0$
high enough such that the source-channel passband $\cF_0$ includes
all frequencies, and then we have for all $\tSNR\geq\tSNR_0$:
\[D(\ej) \leq \frac{1}{1-\delta} \cdot
\frac{\tSNR_0}{\tSNR} \theta_S \] resulting in \[\SDR \geq
(1-\delta) \frac{\tSNR}{\tSNR_0} \SDR_0 =
(1-\delta) \frac{\tSNR}{\tSNR_0} \Gamma_S \Gamma_C
(1+\tSNR_0) = (1-\delta)\SDR^{opt} \] where the equalities are due to
Proposition~\ref{RDeqC}. If the spectra are not bounded, then we artificially
set the pre-filters to be $1$ outside their respective bands (and apply an additional gain in order to comply with the power constraint). This inflicts an arbitrarily small 
SDR loss at $\tSNR_0$, but retains $\SDR\propto\tSNR$, thus the gap from optimality can be kept arbitrarily small. 
\end{proof}

Alternatively, we could prove this result using a the zero-forcing
scheme of \figref{fig_high_resolution}. In fact,
using such a scheme, an even stronger result can be proven: not only
can the encoder be SNR-independent, but so can the decoder.

\subsection{BW Expansion and Compression}
\label{Sub_BW_Change}

We go back now to the cases of bandwidth expansion and compression
discussed at the end of \secref{Sec_Matching}. In these cases, we
can no longer have a single Analog Matching encoder which is
universal for different SNRs, even in the high SNR limit. For
bandwidth expansion ($\rho>1$), the reason is that the source is
perfectly predictable, thus at the limit of high SNR we have that
\[ \var\{U_n|V_{n-1},V_{n-2},\ldots\} \rightarrow \var\{U_n|U_{n-1},U_{n-2},\ldots\}
= 0 \ \ , \] thus the optimum $\beta$ goes to infinity. Any $\beta$
value chosen to ensure correct decoding at some finite SNR, will
impose unbounded loss as the SNR further grows. For bandwidth
compression, the reason is that using any channel predictor suitable
for some finite SNR, we have in the equivalent noise $\tilde Z_n=\tilde Y_n - \tilde X_n$
some component which depends on the channel input (set by the dither). As the SNR
further grows, this component does not decrease, inflicting again
unbounded loss.

By straightforward substitution in Lemma~\ref{lemma_universal}, we
arrive at the following.

 \vspace{5mm} \begin{cor} \label{cor_universal_expansion}
Assume white source and AWGN channel where we are allowed $\rho$
channel uses per source sample. Then using an optimum Analog
Matching encoder for signal to noise ratio $\SNR_0$ and a
suitable (SNR-dependent) decoder, it is possible to approach for any $\tSNR \geq \tSNR_0$:
\beq{D_universal_expansion} \frac{1}{\SDR} =
\frac{1-\min(1,\rho)}{\Bigl(1+\tSNR_0\Bigr)^\rho} +
\frac{\min(1,\rho)}{1+\Phi_\rho(\tSNR,\tSNR_0)} \ \ , \eeq where
\vspace{-3mm} \beq{Phi_rho} \Phi_\rho(\tSNR,\tSNR_0) \Ddef
\frac{1+\tSNR}{1+\tSNR_0} \Bigl[\left(1+\tSNR_0\right)^\rho -
1\Bigr] \  \ . \eeq 
\vspace{3mm}
\end{cor}

Note that the choice of filters in the SNR-dependent decoder remains
simple in this case: For $\rho>1$ the channel post-filter is flat
while the source post-filter is an ideal low-pass filter, while for
$\rho<1$ it is vice versa. The only parameters which change with
SNR, are the scalar filter gains.

\begin{figure}[t]    \label{expansion_fig}
{\centering \subfloat[$\rho=2$] {    \leavevmode
    \epsfig{file=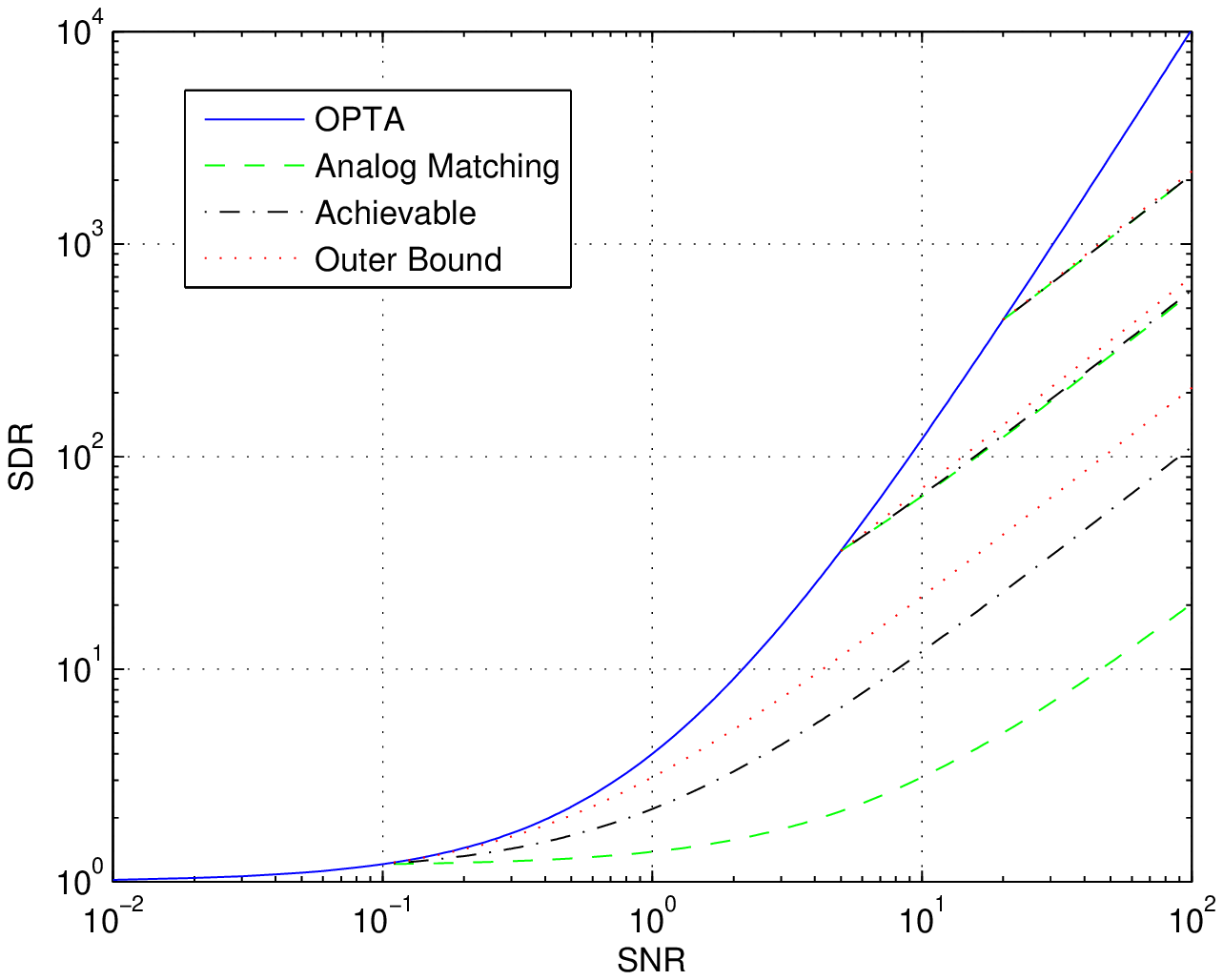, scale = 0.5}}
\
      \subfloat[$\rho=\frac{1}{2}$]{
          \leavevmode
      \epsfig{file=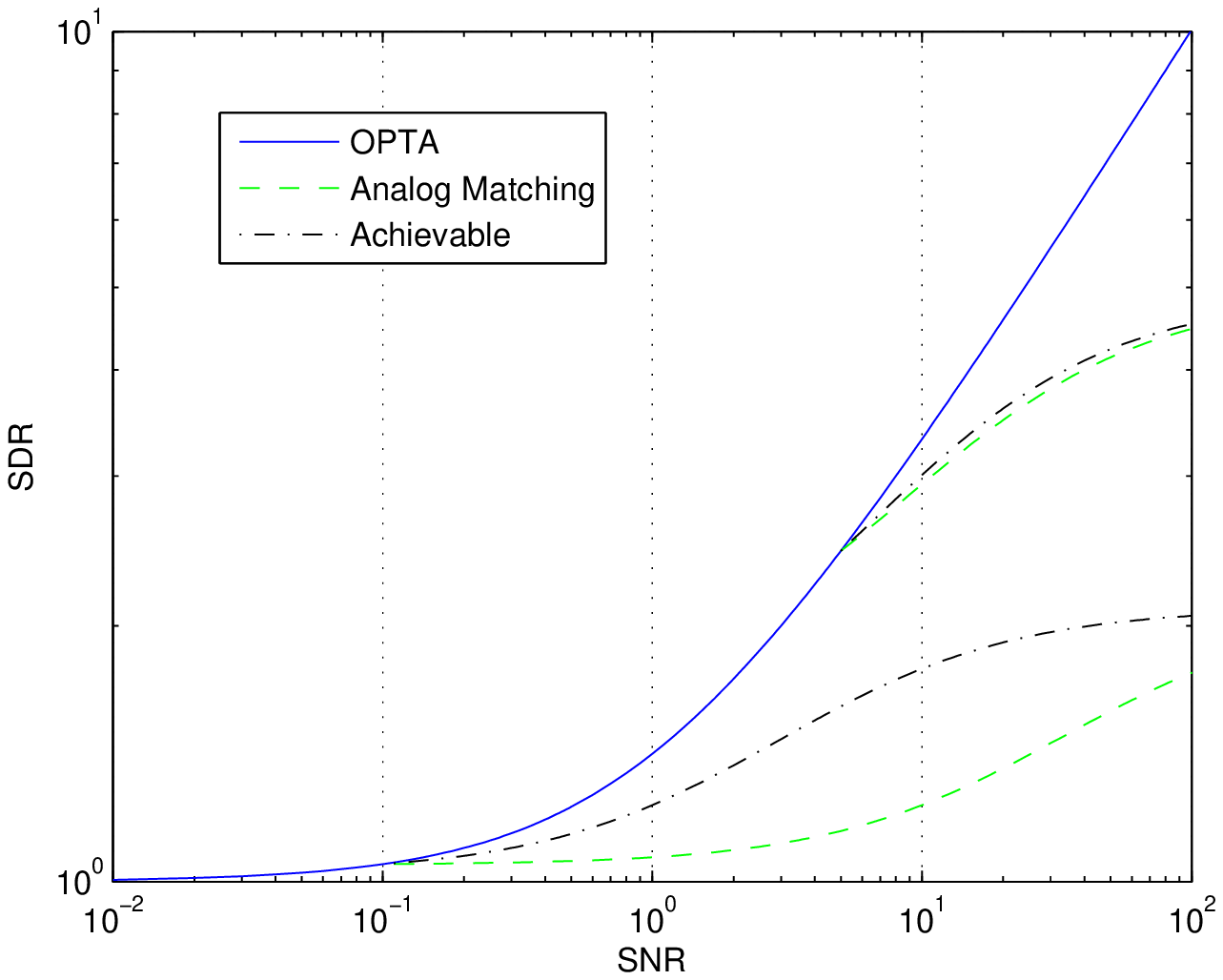, scale = 0.5}}}
\ccaption{Unknown SNR performance: BW expansion and compression. The best known achievable performance, brought for comparison, is due to  \cite{MittalPhamdo,TzvikaBroadcast}.}
\end{figure}

\textbf{Comparison of performance:} In comparison, the performance
reported by different methods in \cite{MittalPhamdo,TzvikaBroadcast}
for these cases has, in terms of \eqref{D_universal_expansion}:
\beq{TzvikaInner} \Phi_\rho(\tSNR,\tSNR_0) = (1+\tSNR)\cdot
(1+\tSNR_0)^{\rho - 1} - 1 \ \, \eeq while \cite{TzvikaBroadcast}
also proves an outer bound for BW expansion ($\rho>1$) on any scheme
which is optimal at some SNR: \beq{TzvikaOuter}
\Phi_\rho(\tSNR,\tSNR_0) = \frac{\tSNR}{\tSNR_0}
\Bigl[(1+\tSNR_0)^\rho-1\Bigr] \ \ . \eeq In both BW expansion and
compression, the Analog Matching scheme does not perform as good as
the previously reported schemes, although the difference vanishes
for high SNR. The basic drawback of analog matching compared to
methods developed specifically for these special cases seems to be,
that these methods apply different ``zooming'' to different source
or channel frequency bands, analog matching uses the same ``zooming
factor'' $\beta$ for all bands. Enhancements to the scheme, such as
the combination of analog matching with pure analog transmission,
may improve these results. \figref{expansion_fig} demonstrates these
results, for systems which are optimal at different SNR levels.

At high SNR, the performance of all these methods and of the outer
bound converge to: \beq{D_universal_expansion_high} \frac{1}{\SDR} =
\frac{1-\min(\rho,1)}{\SNR_0^\rho} +
\frac{\min(\rho,1)}{\SNR\cdot\SNR_0^{\rho-1}} \ \ . \eeq Thus the
Analog Matching scheme, as well as the schemes of
\cite{MittalPhamdo,TzvikaBroadcast}, are all asymptotically optimal
for high SNR among the schemes which achieve $\SDR^{opt}$ at some
SNR.

\subsection{Asymptotic Behavior with BW Change} \label{Sub_High}

Finally we turn back to the general case of non-white spectra with
any $\rho$, and examine it in the high-SNR regime. As in
\secref{Sub_Optimal}, we assume that the channel ISI filter is
known, corresponding with an equivalent noise spectrum $S_Z(\ej)$
known up to a scalar factor.

In the high-SNR limit, Lemma~\ref{lemma_universal} implies:
\beq{D_universal_expansion_gains} \frac{1}{\SDR} =
\left[\frac{1-\min(\rho,1)}{\SNR_0^\rho} +
\frac{\min(\rho,1)}{\SNR\cdot\SNR_0^{\rho-1}}\right] \Gamma_C
\Gamma_S \ \ . \eeq Comparing with
\eqref{D_universal_expansion_high}, we see that the color of the
source and of the noise determines a constant factor by which the
SDR is multiplied, but the dependence upon the SNR remains similar
to the white BW expansion/compression case. The following definition
formalizes this behavior (see \cite{LanemanDiversity}).

\vspace{5mm}
\begin{definition} \label{slope_def}
The \emph{distortion slope} of a continuum of SNR-dependent schemes
is : \beq{lambda} \lambda \Ddef \lim_{\tSNR \rightarrow \infty}
\frac{\log \SDR}{\log \tSNR}\ \eeq where $\SDR$ is the signal to
distortion attained at signal to noise ratio $\tSNR$, where the
limit is taken for a fixed channel filter with noise variance
approaching $0$. \vspace{3mm}
\end{definition}

We use the notation $\lambda=\lambda(\rho)$ in order to emphasize
the dependance of the asymptotic slope upon the bandwidth expansion
factor. The following follows directly from Proposition~\ref{RDeqC}.

\vspace{5mm}\begin{prop} For any source and channel spectra with BW
ratio $\rho$, and for a continuum of schemes achieving the OPTA
performance \eqref{OPTA},
\[ \lambda(\rho) = \rho \ \ . \]
\end{prop}\vspace{5mm}

As for an analog matching scheme which is optimal for a single SNR,
\eqref{D_universal_expansion_gains} implies: \vspace{5mm}
\begin{cor}\label{cor_Slope} For any source and channel spectra and
for a single analog-matching encoder,
\[ \lambda(\rho) = \left\{
\begin{array}{lll}
1 , & \mbox{if $\rho\geq 1$} \\
0 , & \mbox{otherwise}
\end{array}
\right\} \] is achievable.
\end{cor}\vspace{5mm}

This asymptotic slope agrees with the outer bound of
\cite{TzvikaBroadcast} for the (white) bandwidth expansion problem.
For the bandwidth compression problem, no outer bound is known, but
we are not aware of any proposed scheme with a non-zero asymptotic
slope. We believe this to be true for all spectra:

\vspace{5mm}\begin{conj}\label{conj_Slope} For any source and
channel spectra of BW ratio $\rho$, no single encoder which
satisfies \eqref{OPTA} at some $\tSNR_0$ can have a better slope
than that of Corollary \ref{cor_Slope}.
\end{conj}\vspace{5mm}

By this conjecture, the analog matching encoder is asymptotically
optimal among all encoders ideally matched to one SNR. It should be
noted, that schemes which do not satisfy optimality at one SNR
\emph{can} in fact approach the ideal slope $\lambda(\rho)=\rho$.
See e.g. approaches for integer $\rho$ such as bit interleaving
\cite{TaherzadehKhandaniNice}.

\section{Conclusion: Implementation and Applications} \label{conclusion}

We presented the Analog Matching scheme, which optimally transmits a
Gaussian source of any spectrum over a Gaussian channel of any
spectrum, without resorting to any data-bearing code. We showed the
advantage of such a scheme over a separation-based solution, in the
sense of robustness for unknown channel SNR.

The analysis we provided was asymptotic, in the sense that a
high-dimensional lattice is needed. However, 
unlike digital
transmission (and hybrid digital-analog schemes) where reduction of the code block length has a severe
impact on performance, the semi-analog approach offers a potential advantage in terms of block-length. An asymptotic figure
of merit where we expect this advantage to be revealed, is the excess-distortion exponent. Furthermore, the
modulo-lattice framework allows in
practice reduction to low-dimensional, even \emph{scalar} lattices,
with bounded loss.

One approach for scalar implementation of the Analog Matching
scheme, uses \emph{companding} \cite{ItaiThesis}. In this approach,
the scalar zooming factor $\beta$ is replaced by a non-linear
function which compresses the unbounded Gaussian source into a
finite range, an operation which is reverted at the decoder. There
is a problem here, since the entity which needs to be compressed is
actually the innovations process $\tilde Q_n$, unknown at the
encoder since it depends on the channel noise. This can be solved by
compressing $Q_n$, the innovations of the source itself; The effect
of this ``companding encoder-decoder mismatch'' vanishes in the
high-SNR limit. An altogether different approach, is to avoid
instantaneous decoding of the lattice; Instead, the decoder may at
each instance calculate the source prediction using several
hypothesis in parallel. The ambiguity will be solved in future instances,
possibly by a trellis-like algorithm.

In terms of delay, the AM scheme has an additional advantage over previously suggested HDA schemes. It is well known that time-domain approaches have a delay advantage over frequency-domain one, in both source and channel coding. A fully-causal DPCM, for example, can approach the RDF while only using causal filters, on the high-resolution limit. A sub-band coding scheme, in contrast, would have to use a delay-consuming DFT block; see e.g. \cite{Jayant84}. 

Finally, we remark that the AM scheme has further applications. It possesses the basic property, that it converts any colored channel to an
equivalent additive white noise channel of the same capacity as the
original channel, but of the source bandwidth. In the limit of
high-SNR, this equivalent noise becomes Gaussian and independent of
any encoder signal. This property is plausible in multi-user source,
channel and joint source/channel problems, in the presence of
bandwidth mismatch. Applications include computation over MACs
\cite{NazerGastpar07}, multi-sensor detection
\cite{ComputeForwardISIT} and transmission over the parallel relay
network \cite{RematchForwardISIT}.

\appendix

\subsection{Proof of Proposition~\ref{prop_lattice}}

By \cite[(200)]{UriRamiAWGN}, for each of the components $Z_l$:

\beq{Z_l_pdf} \frac{1}{n} \log \frac{f_{Z_l}(z)}{f_{Z'_l(z)}} \leq \epsilon(\Lambda_K) \eeq

where $f(\cdot)$ denotes a probability density function (pdf), $Z'_l$ is AWGN with the same variance as $Z_l$, and $\epsilon(\Lambda_K) \rightarrow 0$ as $K\rightarrow\infty$ for a sequence of lattices which is Rogers-good (i.e. lattices for which volume of the covering sphere approaches that of the Voronoi cell). 
Now assume without loss of generality that $\alpha_l^2$ is a non-increasing for $l>0$, and for some fixed $\delta$ let $L'$ be the minimal index such that \[ \sum_{l=L'+1}^\infty \alpha_l^2 \leq \delta \ \ . \]
Let $Z_\delta = \sum_{l=1}^{L'} \alpha_l Z_l$. Using \eqref{Z_l_pdf} and convolution of pdfs,

\[  \frac{1}{n} \log \frac{f_{Z_\delta}(z)}{f_{Z'_\delta(z)}} \leq \epsilon^{L'}(\Lambda_K) \ \ , \]
where $Z'_\delta$ is AWGN with the same variance as $Z_\delta$. Since $\epsilon^{L'}$ approaches zero as a function of $K$, 

\[ \lim_{K\rightarrow\infty} \Pr\{Z_0+Z_\delta \notin \Nu_K\} = 0 \ \ , \]  for lattices which are good for AWGN coding.

We are left with the ``tail'' $\tilde Z = \sum_{l=L'+1}^\infty \alpha_l Z_l$, which has variance $\delta$. By continuity arguments, \[ \lim_{\delta \rightarrow 0^+} \Pr \{Z_0 + Z_\delta + \tilde Z \notin \Nu_K  | Z_0 + Z_\delta \in \Nu_K \}  = 0 \ \ . \]

The result follows now by standard arguments of taking $\epsilon$ and $\delta$ to zero simultaneously.  We have assumed the use of a sequence of lattices that is simultaneously Rogers-good and AWGN-good. By \cite{GoodLattices}, such a sequence indeed exists.

\subsection{Proof of Lemma~\ref{lemma_AM1}}

By the properties of the modulo-lattice operation,
$\tilde X_n$ is a white process. Now the channel from $\tilde X_n$ to $\tilde Y_n$ is
identical to the channel of \eqref{Pred_C}, thus we
have that:
\[Y'_n = (\tilde X_n + \tilde Z_n)*(\delta_n-{p_C}_n) = \tilde X_n + I_n +
Z''_n \ \ , \] where $\tilde Z_n$ has spectrum $S_{\tilde Z}(\ej)$ \eqref{S_tildeZ}, and consequently $Z''_n = \tilde Z_n*(\delta_n-{p_C}_n)$
is its white prediction error, with variance $\frac{1-\alpha}{\alpha}
\theta_C$ according to \eqref{channel_exact}. Now since $\tilde Y_n
= Y'_n - I_n$ is the optimum linear estimator for $\tilde X_n$ from the
channel output, the orthogonality principle dictates that the
estimation error is uncorrelated with the process $Y'_n$, resulting
in an additive backward channel (see e.g.
\cite{ZamirKochmanErezDPCM-IT}):
\[\tilde X_n = Y'_n - I_n + Z''_n \ \ . \] Switching back to
a forward channel, we have \[ Y'_n = \alpha (\tilde X_n + Z'_n) +
I_n \ \ ,\] where $Z'_n$ is white with the same variance as $Z''_n$. Furthermore, since $Z'_n$ is a function of the processes $\{\tilde X_n\}$ and $\{Z_n\}$, it is independent of all $U_n$.

\subsection{Proof of Lemma~\ref{lemma_AM}}

By the properties of the modulo-lattice operation, \[T_n 
 = \beta (U_n - J_n) + {Z_{eq}}_n \ \ , \] resulting in the equivalent channel of \figref{finite_k_matching_fig}. By the correct initialization assumption, \eqref{ideal_additive} holds for all past instances, thus $\bT_n$ is a combination noise (see Definition~\ref{def_Combination}). In light of Proposition~\ref{prop_lattice}, it is only left to show that the variance of $T_n$ is strictly less than the lattice second moment $\theta_C$. To that end, note that under the correct initialization assumption, the past samples of the process $V_n$ indeed behave as samples of a stationary process of spectrum $S_V(\ej)$ \eqref{S_V}, for which 
$P_S(\ej)$ is the optimal predictor. It follows that $U_n-J_n$ is white, with variance \beqn{long_in_proof}
\var\{U_n-J_n\} &=& \var\{U_n | V_{n-1},V_{n-2},\ldots\} \nonumber
\\ &\stackrel{(a)}=& P_e \Bigl(S_U + \frac{\var\{{Z_{eq}}_n\}}{\beta^2}\Bigr) -
\frac{\var\{{Z_{eq}}_n\}}{\beta^2} \nonumber
\\ &=& P_e \Bigl(S_U +
\frac{\beta_0^2}{\beta^2} \theta_S \Bigr) -
\frac{\beta_0^2}{\beta^2} \theta_S \nonumber
\\ &<& \frac{\beta_0^2}{\beta^2}
\Bigl[P_e(S_U+\theta_S)-\theta_S\Bigr] \nonumber
\\ &\stackrel{(b)}=& \frac{\beta_0^2}{\beta^2}
\cdot \frac {\alpha}{1-\alpha} \theta_S \nonumber
\\ &=& \frac{\alpha
\theta_C}{\beta^2} \ \ , \nonumber \eeqn where $(a)$ holds by
\eqref{noisy_prediction}, and $(b)$ holds by applying the same in the
opposite direction, combined with \eqref{source_exact}. By the
whiteness of ${Z_{eq}}_n$ and its independence of all $U_n$, we have
that $U_n-J_n$ is independent of ${Z_{eq}}_n$, thus the variance of
$T_n$ is given by \beq{AM_Condition1} \var\{T_n\} = \beta^2
\var\{U_n-J_n\} + \var\{{Z_{eq}}_n\} < \theta_C \ \ . \eeq The
margin from $\theta_C$ depends on the margin in the inequality in
the chain above, which depends only on $S_U(\ej)$, $\theta_C$ and
$\beta$, and is strictly positive for all $\beta<\beta_0$.

\subsection{Proof of Lemma~\ref{lemma_universal}}

We work with the optimum Analog Matching encoder for the noise
spectrum $S_{Z0}(\ej)$. At the decoder, we note that for any choice
of the channel post-filter $G_2(\ej)$, we have that the equivalent
noise $\nk{Z_{eq}}$ is the noise $\nk{\tilde Z}\Ddef \nk{\tilde Y} -
\nk{\tilde X}$ passed through the filter $1-P_C(\ej)$. Consequently,
this noise has spectrum: \[ S_{eq}(\ej) = S_{\tilde Z}(\ej)
|1-P_C(\ej)|^2 \ \ . \] The filter $G_2(\ej)$ should, therefore, be
the Wiener filter which minimizes $S_{\tilde Z}(\ej)$ at each
frequency. This filter achieves a noise spectrum \[S_{\tilde Z}(\ej)
= \frac{\theta_C-S_{Z0}(\ej)}{\theta_C-S_{Z0}(\ej)+S_Z(\ej)}
S_Z(\ej) \] inside $\cF_C$, and $\theta_C$ outside. Denoting the
variance of the (white) equivalent noise in the case
$S_{Z0}(\ej)=S_Z(\ej)$ as $S_0=(1-\alpha)\theta_C$ \eqref{var_Z_eq}, we find
that:
\[|1-P_C(\ej)|^2 = \frac{S_0
\theta_C}{(\theta_C-S_{Z0}(\ej)) S_{Z0}(\ej)} \] inside $\cF_C$, and
$S_0/\theta_C$ outside. We conclude that we have
equivalent channel noise with spectrum
\beq{S_eq_robust} S_{eq}(\ej) = \frac{S_Z(\ej)}{S_{Z0}(\ej)} \cdot
\frac{\theta_C}{\theta_C-S_{Z0}(\ej)+S_Z(\ej)} S_0 =
\frac{S_S(\ej)-\theta_S}{\Phi(\ej) \theta_S} S_0 \eeq inside
$\cF_C$, and $S_0$ outside. Now, since this spectrum is
everywhere upper-bounded by $S_0$, we need not worry about
correct decoding. The source post-filter input is the
source, corrupted by an additive noise ${Z_{eq}}_n/\beta$,
with spectrum arbitrarily close to
\[\frac{S_{eq}(\ej)}{\beta_0^2} = \frac{S_S(\ej)-\theta_S}{\Phi(\ej)} \ \  \] inside
$\cF_C$, and $\theta_S$ outside. Now again we face optimal linear
filtering, and we replace the source post-filter $F_2(\ej)$ by the
Wiener filter for the source, to arrive at the desired result.
\bibliography{../Latex/mybib}

\begin{thebibliography}{10}

\bibitem{BergerOptimalPAM}
T.~Berger and D.W. Tufts.
\newblock Optimum pulse amplitude modulation part {I}: {T}ransmitter-receiver
  design and bounds from information theory.
\newblock {\em IEEE Trans. Info. Theory}, IT-13:196--208, Apr. 1967.

\bibitem{ChenWornell98}
B.~Chen and G.~Wornell.
\newblock Analog error-correcting codes based on chaotic dynamical systems.
\newblock {\em IEEE Trans. Communications}, 46:881--890, July 1998.

\bibitem{CDEF-MMSE-DFE}
J.M. Cioffi, G.P. Dudevoir, M.V. Eyuboglu, and G.D. J.~Forney.
\newblock {MMSE} decision-feedback equalizers and coding - {P}art {I}:
  {E}qualization results.
\newblock {\em IEEE Trans. Communications}, COM-43:2582--2594, Oct. 1995.

\bibitem{Costa83}
M.H.M. Costa.
\newblock Writing on dirty paper.
\newblock {\em IEEE Trans. Info. Theory}, IT-29:439--441, May 1983.

\bibitem{CouchBook}
L.~W. Couch.
\newblock {\em Digital \& Analog Communication Systems (7th Edition)}.
\newblock Prentice Hall, 2006.

\bibitem{CoverBook}
T.~M. Cover and J.~A. Thomas.
\newblock {\em Elements of Information Theory}.
\newblock Wiley, New York, 1991.

\bibitem{GoodLattices}
U.~Erez, S.~Litsyn, and R.~Zamir.
\newblock Lattices which are good for (almost) everything.
\newblock {\em IEEE Trans. Info. Theory}, IT-51:3401--3416, Oct. 2005.

\bibitem{UriRamiAWGN}
U.~Erez and R.~Zamir.
\newblock Achieving 1/2 log(1+{SNR}) on the {AWGN} channel with lattice
  encoding and decoding.
\newblock {\em IEEE Trans. Info. Theory}, IT-50:2293--2314, Oct. 2004.

\bibitem{Forneyallerton04}
G.~D. Forney, Jr.
\newblock {S}hannon meets {W}iener {II}: On {MMSE} estimation in successive
  decoding schemes.
\newblock In {\em 42nd Annual Allerton Conference on Communication, Control,
  and Computing, Allerton House, Monticello, Illinois}, Oct. 2004.

\bibitem{ToCode}
M.~Gastpar, B.~Rimoldi, and M.~Vetterli.
\newblock To code or not to code: {L}ossy source-channel communication
  revisited.
\newblock {\em IEEE Trans. Info. Theory}, IT-49:1147--1158, May 2003.

\bibitem{Goblick65}
T.J. Goblick.
\newblock Theoretical limitations on the transmission of data from analog
  sources.
\newblock {\em IEEE Trans. Info. Theory}, IT-11:558--567, 1965.

\bibitem{GuessVaranasiIT}
T.~Guess and M.~Varanasi.
\newblock An information-theoretic framework for deriving canonical
  decision-feedback receivers in {G}aussian channels.
\newblock {\em IEEE Trans. Info. Theory}, IT-51:173--187, Jan. 2005.

\bibitem{Jayant84}
N.~S. Jayant and P.~Noll.
\newblock {\em Digital Coding of Waveforms}.
\newblock Prentice-Hall, Englewood Cliffs, NJ, 1984.

\bibitem{RematchForwardISIT}
Y.~Kochman, A.~Khina, U.~Erez, and R.~Zamir.
\newblock Rematch and forward for parallel relay networks.
\newblock In {\em ISIT-2008, Toronto, ON}, pages 767--771, 2008.

\bibitem{JointWZ-WDP}
Y.~Kochman and R.~Zamir.
\newblock Joint {W}yner-{Z}iv/dirty-paper coding by modulo-lattice modulation.
\newblock {\em IEEE Trans. Info. Theory}, IT-55:4878--4899, Nov. 2009.

\bibitem{LanemanDiversity}
J.N. Laneman, E.~Martinian, G.W. Wornell, and J.G. Apostolopoulos.
\newblock Source-channel diversity approaches for multimedia communication.
\newblock {\em IEEE Trans. Info. Theory}, IT-51:3518--3539, Oct. 2005.

\bibitem{ItaiThesis}
I.~Leibowitz.
\newblock The {Z}iv-{Z}akai bound at high fidelity, analog matching, and
  companding.
\newblock Master's thesis, Tel Aviv University, Nov. 2007.

\bibitem{MittalPhamdo}
U.~Mittal and N.~Phamdo.
\newblock Hybrid digital-analog ({HDA}) joint source-channel codes for
  broadcasting and robust communications.
\newblock {\em IEEE Trans. Info. Theory}, IT-48:1082--1103, May 2002.

\bibitem{NarayananCaireReport}
K.~Narayanan, M.~P. Wilson, and G.~Caire.
\newblock Hybrid digital and analog {C}osta coding and broadcasting with
  bandwidth compression.
\newblock Technical Report 06-107, Texas A\&M University, College Station,
  August 2006.

\bibitem{ComputeForwardISIT}
B.~Nazer and M.~Gastpar.
\newblock Compute-and-forward: Harnessing interference with structured codes.
\newblock In {\em ISIT-2008, Toronto, ON}, pages 772--776, 2008.

\bibitem{NazerGastpar07}
B.~Nazer and M.~Gastpar.
\newblock Computation over multiple-access channels.
\newblock {\em IEEE Trans. Info. Theory}, IT-53:3498--3516, Oct. 2007.

\bibitem{Puri2005}
V.~M. Prabhakaran, R.~Puri, and K.~Ramchandran.
\newblock A hybrid analog-digital framework for source-channel broadcast.
\newblock In {\em Proceedings of the 43rd Annual Allerton Conference on
  Communication, Control and Computing}, 2005.

\bibitem{PuriRamchandranPRISM}
R.~Puri and K.~Ramchandran.
\newblock {PRISM}: {A} 'reversed' multimedia coding paradigm.
\newblock In {\em Proc. IEEE Int. Conf. Image Processing, Barcelona}, 2003.

\bibitem{TzvikaBroadcast}
Z.~Reznic, M.~Feder, and R.~Zamir.
\newblock Distortion bounds for broadcasting with bandwidth expansion.
\newblock {\em IEEE Trans. Info. Theory}, IT-52:3778--3788, Aug. 2006.

\bibitem{ShamaiVerduZamir}
S.~Shamai, S.~Verd\'u, and R.~Zamir.
\newblock Systematic {L}ossy {S}ource/{C}hannel {C}oding.
\newblock {\em IEEE Trans. Info. Theory}, 44:564--579, March 1998.

\bibitem{TaherzadehKhandaniNice}
M.~Taherzadeh and A.~K. Khandani.
\newblock Robust joint source-channel coding for delay-limited applications.
\newblock In {\em ISIT-2007, Nice, France}, pages 726--730, 2007.

\bibitem{Tomlinson}
M.~Tomlinson.
\newblock New automatic equalizer employing modulo arithmetic.
\newblock {\em Elect. Letters}, 7:138--139, March 1971.

\bibitem{VanTrees68}
H.~L.~Van Trees.
\newblock {\em Detection, Estimation, and Modulation theory}.
\newblock Wiley, New York, 1968.

\bibitem{TrottITW96}
M.~D. Trott.
\newblock Unequal error protection codes: Theory and practice.
\newblock In {\em Proc. of Info. Th. Workshop, Haifa, Israel}, page~11, June
  1996.

\bibitem{Vaishampayan2003}
V.A. Vaishampayan and S.I.R. Costa.
\newblock Curves on a sphere, shift-map dynamics, and error control for
  continuous alphabet sources.
\newblock {\em IEEE Trans. Info. Theory}, IT-49:1658--1672, July 2003.

\bibitem{WilsonNarayananCaire07}
M.P Wilson, K.~Narayanan, and G.~Caire.
\newblock Joint source chennal coding with side information using hybrid
  digital analog codes.
\newblock In {\em Proceedings of the Information Theory Workshop, Lake Tahoe,
  CA}, pages 299--308, Sep. 2007.

\bibitem{WynerZiv76}
A.D. Wyner and J.~Ziv.
\newblock The rate-distortion function for source coding with side information
  at the decoder.
\newblock {\em IEEE Trans. Info. Theory}, IT-22:1--10, Jan., 1976.

\bibitem{FederZamirLQN}
R.~Zamir and M.~Feder.
\newblock On lattice quantization noise.
\newblock {\em IEEE Trans. Info. Theory}, pages 1152--1159, July 1996.

\bibitem{FederZamir94}
R.~Zamir and M.~Feder.
\newblock Information rates of pre/post filtered dithered quantizers.
\newblock {\em IEEE Trans. Info. Theory}, pages 1340--1353, Sep. 1996.

\bibitem{ZamirKochmanErezDPCM-IT}
R.~Zamir, Y.~Kochman, and U.~Erez.
\newblock Achieving the {G}aussian rate distortion function by prediction.
\newblock {\em IEEE Trans. Info. Theory}, IT-54:3354--3364, July 2008.

\bibitem{RamiShamaiUriLattices}
R.~Zamir, S.~Shamai, and U.~Erez.
\newblock Nested linear/lattice codes for structured multiterminal binning.
\newblock {\em IEEE Trans. Info. Theory}, IT-48:1250--1276, June 2002.

\bibitem{Ziv70}
J.~Ziv.
\newblock The behavior of analog communication systems.
\newblock {\em IEEE Trans. Info. Theory}, IT-16:587--594, 1970.

\end{thebibliography}
\bibliographystyle{plain}

\end{document}